\pdfoutput=1

\documentclass[9pt,sigconf,anonymous=false,review=false]{acmart}

\usepackage{amsmath}
\usepackage{mathtools}
\usepackage[capitalise]{cleveref}
\usepackage{amsthm}
\usepackage{subcaption}

\setcopyright{none}
\renewcommand\footnotetextcopyrightpermission[1]{} 
\usepackage{tikz}
\usetikzlibrary{automata,arrows,positioning,calc,decorations}
\usetikzlibrary{shapes,shapes.geometric,arrows,fit,calc,positioning,automata}
\usetikzlibrary{patterns,decorations.pathreplacing,decorations.markings}

\newcommand{\Set}[1]{\{#1\}}

\newcommand{\cgreencolorA}[2]{\text{\textcolor{green}{$\mathcal{I}_{#1}^{#2}$}}}
\newcommand{\cbluecolorA}[2]{\text{\textcolor{blue}{$\mathcal{A}_{#1}^{#2}$}}}
\newcommand{\credcolorA}[2]{\text{\textcolor{red}{$\mathcal{A}_{#1}^{#2}$}}}

\newcommand{\rgreencolorA}[2]{\text{\textcolor{green}{$\mathcal{A}_{#1}^{#2}$}}}
\newcommand{\rbluecolorA}[2]{\text{\textcolor{blue}{$\mathcal{A}_{#1}^{#2}$}}}
\newcommand{\rredcolorA}[2]{\text{\textcolor{red}{$\mathcal{A}_{#1}^{#2}$}}}
\newcommand{\rgreencolorI}[2]{\text{\textcolor{green}{$\mathcal{I}_{#1}^{#2}$}}}
\newcommand{\rredcolorI}[2]{\text{\textcolor{red}{$\mathcal{I}_{#1}^{#2}$}}}

\newtheorem*{claim}{Claim}
\newtheorem{remark}{Remark}
\newcommand\restr[2]{\ensuremath{\left.#1\right|_{#2}}}

\AtBeginDocument{%
  \providecommand\BibTeX{{%
    \normalfont B\kern-0.5em{\scshape i\kern-0.25em b}\kern-0.8em\TeX}}}

\begin{document}

\title{Continuous Pushdown VASS in One Dimension are Easy}

\author{Guillermo A. P\'erez}
\email{guillermo.perez@uantwerpen.be}
\orcid{0000-0002-1200-4952}
\author{Shrisha Rao}
\orcid{0000-0001-5559-7287}
\email{shrisha.rao@uantwerpen.be}
\affiliation{%
  \institution{University of Antwerp}
  \streetaddress{Middelheimlaan 1}
  \city{Antwerp}
  \country{Belgium}
  \postcode{2020}
}



\begin{abstract}
    A pushdown vector addition system with states (PVASS) extends the model of vector addition systems with a pushdown stack. The algorithmic analysis of PVASS has applications such as static analysis of recursive programs manipulating integer variables. Unfortunately, reachability analysis, even for one-dimensional PVASS is not known to be decidable. We relax the model of one-dimensional PVASS to make the counter updates continuous and show that in this case reachability, coverability, and boundedness are decidable in polynomial time. In addition, for the extension of the model with lower-bound guards on the states, we show that coverability and reachability are in NP, and boundedness is in coNP.
\end{abstract}






\maketitle

\section{Introduction}
Vector addition systems with states (VASS) are commonly used to model
distributed systems and concurrent systems with integer variables. A VASS
consists of a set of (control) states and a set of counters. Transitions
between states are labelled with vectors of integers (usually encoded in
binary) that are added to the current values of the counters. Importantly,
transitions that would result in a counter value becoming negative are
disallowed. 

An equivalent way of understanding the model is to see the
counters as unary-alphabet stacks. This alternative formulation has a natural
extension obtained by adding one general stack (i.e. its alphabet is not a
singleton) to it. Pushdown VASS (PVASS), as they are usually called, can be used to
model recursive programs that manipulate integer variables. Arguably the most
basic question one can attempt to answer algorithmically in a
computational model is that of \emph{reachability}. In the context of
(pushdown) VASS, we ask whether a given target configuration (formed by the
current state and the values of the counters) can be seen along a run from a
given source configuration. While the complexity of reachability for VASS is
now better
understood~\cite{DBLP:conf/lics/LerouxS19,DBLP:journals/jacm/CzerwinskiLLLM21},
for PVASS it is not known to be decidable and the best known lower
bound is \textsc{HyperAck}-hardness~\cite{DBLP:conf/csl/LerouxPS14}. In one dimension, the problem is also not known to be decidable and the
known lower bound is \textsc{Pspace}-hardness~\cite{englert2021lower}.

Motivated by the (complexity) gap in our understanding of reachability for
PVASS, researchers have studied the problem for different relaxations of the
model: A PVASS is \emph{bidirected}~\cite{DBLP:conf/icalp/GanardiMPSZ22} if the effect (on the stack and the
counters) of every transition can be (immediately) reversed; A
$\mathbb{Z}$-PVASS~\cite{DBLP:conf/cav/HagueL11} allows counters to hold negative values; A
\emph{continuous} PVASS~\cite{CPVASS_bala} instead allows them to hold nonnegative values and
counter updates labelling a transition can be scaled by any $\alpha \in (0,1]$
when taking the transition. For all of these, reachability is known to be
decidable. For some of them, lower complexity bounds for the special case of one dimension have also been established. See \Cref{tab:reach} for a summary of known results.

\begin{table*}
  \centering
  \caption{Previously known complexity bounds for the reachability problem in PVASS and
  relaxations thereof}
  \label{tab:reach}
  \begin{tabular}{c|c||c|c|c}
    & PVASS & Bidirected PVASS & $\mathbb{Z}$-PVASS & Continuous PVASS \\
    \hline
    General & \textsc{HyperAck}-hard & ${} \in {}$ \textsc{Ack} &
    \textsc{NP}-complete & \textsc{NEXP}-complete \\
    1 Dimension & \textsc{Pspace}-hard & ${} \in {}$ \textsc{Pspace} &
    \textsc{NP}-complete & ${} \in {}$ \textsc{NEXP}
  \end{tabular}
\end{table*}

In this work, we study reachability, boundedness, and coverability in continuous PVASS in one dimension. The boundedness problem asks whether the set of all reachable configurations, from a given source configuration, is finite. In turn, coverability asks whether a given state can be seen along a run from a given source configuration. In contrast to reachability, coverability is known to be decidable and in \textsc{EXPspace} for PVASS in one dimension~\cite{leroux2015coverability}. Similarly, boundedness is known to be decidable and in \textsc{HyperAck}, this time in general, not only in one dimension~\cite{DBLP:conf/csl/LerouxPS14}.

\paragraph{Contributions.}
In this paper, we prove that, for continuous PVASS in one dimension, reachability, coverability, and boundedness are decidable in \textsc{Ptime}. We further show that if one adds to the model lower-bound guards on the states for the counter (thus allowing for a ``tighter'' relaxation of the original model) then reachability and coverability are in \textsc{NP} while boundedness is in \textsc{coNP}. See \Cref{tab:contrib} for a summary of our contributions.

\begin{table}
    \centering
    \caption{New results for continuous PVASS in one dimension}\label{tab:contrib}
    \begin{tabular}{c|c||c}
         & PVASS & Cont. PVASS (/ with low. bounds) \\
         \hline
         Reach & \textsc{Pspace}-hard & ${}\in{}$ \textsc{Ptime} / \textsc{NP} \\
         Cover & ${} \in {}$ \textsc{EXPspace} & ${}\in{}$ \textsc{Ptime} / \textsc{NP} \\
         Bounded & ${} \in {}$ \textsc{HypAck} & ${}\in{}$ \textsc{Ptime} / \textsc{coNP} \\
    \end{tabular}
\end{table}

\section{Preliminaries}
We first recall a definition of pushdown automata and
then we extend this definition to continuous PVASS.

\subsection{Pushdown automata and Context-free grammars}

\begin{definition}[Pushdown automata]
  A pushdown automaton (PDA, for short) is a tuple
  $\mathcal{P}=(S,\Sigma,\Gamma,\delta,s_0,\bot,F)$ where:
  \begin{itemize}
    \item $S$ is a finite set of states,
    \item $\Sigma$ a finite (possibly empty) alphabet, 
    \item $\Gamma$ a finite stack alphabet, 
    \item $s_0\in S$ the initial state,
    \item $\bot\in\Gamma$ the initial stack symbol, 
    \item $F\subseteq S$ a set of accepting states,
    \item and $\delta:S\times S\rightarrow
      \big(\Sigma\cup\epsilon\big)\times\big(\Set{a,\overline{a}\mid
      a\in\Gamma\backslash\bot}\cup\epsilon\big)$ a partial function,
      where, $a$ and $\overline{a}$ denote pushing and popping $a$ from the
      stack respectively.
  \end{itemize}
\end{definition}

A \textit{configuration} of a PDA $\mathcal{P}$ is of the form $(s,w,\alpha)\in
S\times\Sigma^*\times\Gamma^*$ where $s$ represents the current state of the
PDA, $w$ the word read by the PDA until reaching the state $s$ and $\alpha$
the current stack contents of the PDA (with the right being the ``top'' from
which we pop and onto which we push). The \textit{initial configuration} $q_0$ is $(s_0,\epsilon,\bot)$. 

A \textit{run} of a PDA $\mathcal{P}$ is of the form $\pi=q_0 q_1 \dots q_n$ where
$q_i=(s_i,w_i,\alpha_i)$ is a configuration, for all $0\leq i\leq n$, and the
following hold for
all $0\leq i<n$:
\begin{itemize}
  \item $\delta(s_i,s_{i+1})$ is defined,
  \item $w_{i+1}=w_i\cdot\delta(s_i,s_{i+1})_1$,
  \item $\alpha_{i+1}=\alpha_i$ if
    $\delta(s_i,s_{i+1})_2=\epsilon$,
  \item $\alpha_{i+1}=\alpha_i\cdot a$ if $\delta(s_i,s_{i+1})_2=a$, and
  \item $\alpha_i=\alpha_{i+1}\cdot a$ if $\delta(s_i,s_{i+1})_2=\overline{a}$.
\end{itemize}
Above, $\delta(\_,\_)_i$ represents the $i$-th component of the tuple.
For any $Q \subseteq S$, we say the run reaches $Q$ if $s_n \in Q$.

\subsubsection{Acceptance conditions of a PDA}
There are three classical notions of accepting runs $q_0 \dots q_n$ for PDAs. Below, we recall them writing $q_i =
(s_i,w_i,\alpha_i)$.
\begin{description}
    \item[State reachability] says the run is accepting if $s_n \in
      F$.
    \item[Empty stack] says the run is accepting if $\alpha_n = \bot$,
      i.e. the stack is empty, no matter the state.
    \item[Both] says the run is accepting if and only if both
      conditions above hold true.
\end{description}
Note that $q_0=(s_0,w_0,\alpha_0)=(s_0,\epsilon,\bot)$ means that accepting runs start from the initial configuration.
It is well known that all three acceptance conditions are logspace interreducible (see, e.g. \cite[Supplementary Lecture E]{kozen:automata}). We only focus on state reachability, that is, a run $\pi$ of a PDA is \textit{accepting} if $q_0$ is the initial
configuration and $s_n\in F$. 

The language of a PDA $\mathcal{P}$, denoted by $L(\mathcal{P})$, is the set of all words $w_n \in \Sigma^*$ read by accepting runs $q_0 \dots
(s_n,w_n,\alpha_n)$ of $\mathcal{P}$. The \emph{Parikh image} $\Phi(w)$ of a word $w \in \Sigma^*$, i.e. the vector in $\mathbb{N}^{|\Sigma|}$ such that its $i$\textsuperscript{th} is the number of times the $i$\textsuperscript{th} letter of $\Sigma$ (assuming an arbitrary choice of total order) appears in $w$.

\subsubsection{Context-free grammars} CFGs, for short, are a model that is expressively equivalent to PDAs in terms of their languages. The models are logspace reducible to each other \cite[Section 5.3]{Hopcroft+Ullman/79/Introduction}.

\begin{definition}[Context-free grammars] A CFG is a tuple $G=(V,\Sigma,P,S)$, where $V$ is a set of variables; $\Sigma$, a set of terminals; $P\subset V\times\Set{\rightarrow}\times \Set{V,\Sigma}^*$, a set of productions; and $S\in V$, the start symbol.
\end{definition}

The \textit{production symbol} ``$\rightarrow$'' separates the \textit{head} (a variable) of the production, to the left of $\rightarrow$, from the \textit{body} (a string of variables and terminals) of the production, to the right of $\rightarrow$.
Each variable represents a language, i.e., a (possibly empty) set of strings of terminals. The body of each production represents one way to form strings in the language of the head.

\begin{example}
    The grammar $G=(\Set{A},\Set{a,b},P,S=A)$ represents the set of all palindromes over $\Set{a,b}$ where the productions are:
    \[
    \begin{aligned}
        A & \rightarrow \epsilon,\\
        A & \rightarrow a,
    \end{aligned}
    \quad\quad
    \begin{aligned}
        A & \rightarrow b, \\
        A & \rightarrow aAa,
    \end{aligned}
    \quad\quad
    \begin{aligned}
    A & \rightarrow bAb.\\
    & 
    \end{aligned}
    \]
    The word $abaaaba$, for example, is in the language of $A$ since it can be obtained by $A\rightarrow aAa\rightarrow abAba\rightarrow abaAaba\rightarrow abaaaba$ where the fourth, fifth, again the fourth, and finally, the second production rules are applied, in that order.
\end{example}
    
\textit{Chomsky normal form} (or CNF) \cite[Section 4.5]{Hopcroft+Ullman/79/Introduction} is a normal form
for CFGs with the restriction that all production rules can only be of the form $A \rightarrow BC$, or $A \rightarrow a$, or $S \rightarrow \varepsilon$. Converting a CFG to its Chomsky normal form can lead to at most a cubic explosion in size.

\subsection{Continuous pushdown VASS}
A continuous pushdown vector addition system with states in one dimension is a PDA with
a continuous counter.

\begin{definition}[C1PVASS]
    A continuous pushdown VASS (with lower-bound guards) in one dimension (C1PVASS) is a tuple:
    \[\mathcal{A}=(S,\Sigma,\Gamma,\delta,\bot,s_0,F,\ell)\] where $S$ is a finite set of
    states; $s_0\in S$, the initial state; $F\subseteq S$, a set of accepting
    states; $\Sigma$, a finite alphabet; $\Gamma$, a finite stack alphabet;
    $\bot\in\Gamma$, the initial stack symbol; 
    $\delta:S\times S\rightarrow\big(\Sigma\cup\epsilon\big)\times\mathbb{Z}\times\big(\Set{a,\overline{a}\mid
    a\in\Gamma\backslash\bot}\cup\epsilon\big)$, a partial transition function; and
    $\ell:S\rightarrow\mathbb{N}$, a function that assigns the lower bounds to the states.
\end{definition}

Since we only study runs, and not languages, of C1PVASS, we henceforth omit $\Sigma$. We also assume, without loss of generality, that the set $F$ is a singleton. This can be done by adding a new final state $f'$ to $S$ and adding transitions for all $f\in F$ to $f'$ which read $\epsilon$, have a $+0$ counter update, and do not modify the stack.
With these assumptions, we have a simpler representation of a C1PVASS
\[\mathcal{A}=(S,\Gamma,\delta,\bot,s_0,f,\ell)\]
where $\delta$ is now of the form $\delta:S\times S\rightarrow\mathbb{Z}\times\big(\Set{a,\overline{a}\mid
    a\in\Gamma\backslash\bot}\cup\epsilon\big)$.
    
A \textit{configuration} of a C1PVASS is of the form $(s,\alpha,c)$ where
$s$ and $\alpha$ are as for PDAs, and $c\in\mathbb{R}_{\geq 0}$ is the current
nonnegative value of the counter with the property that $c\geq\ell(s)$, that is, the counter value at a state must be at least the lower bound on that state. The \textit{initial configuration} $q_0$ of
the C1PVASS is $(s_0,\bot,0)$.

A \textit{run} of the C1PVASS $\mathcal{A}$ is a
sequence of configurations $\pi=q_0 q_1 \dots q_n$ with $q_i =
(s_i,\alpha_i,c_i)$ such that $\restr{\pi}{\mathcal{P}_\mathcal{A}}$, obtained by
removing the counter values $c_i$, is a run in the PDA $\mathcal{P}_\mathcal{A}$, obtained by
removing the counter updates from $\mathcal{A}$, and the following holds for all $0\leq
i<n$: $c_{i+1}=c_i+ \gamma \delta(s_i,s_{i+1})_1$ for some $\gamma \in
\mathbb{R} \cap (0,1]$. We call the $\gamma$ \textit{scaling factors}.

\begin{example}
    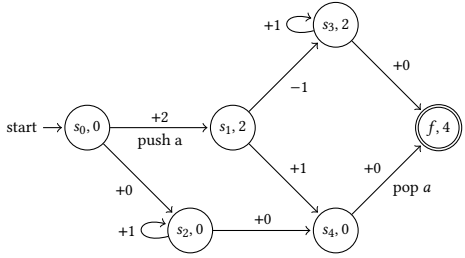
\begin{figure}
        \centering
        \begin{tikzpicture}[shorten >=1pt,auto,node distance=1.9 cm, scale = 0.7, transform shape]

        \node[state,initial](s0)[]{$s_0,0$};
        \node[state](s1)[right=of s0]{$s_1,2$};
        \node[state](s2)[below right=of s0]{$s_2,0$};
        \node[state](s3)[above right=of s1]{$s_3,2$};
        \node[state](s4)[below right=of s1]{$s_4,0$};
        \node[state,accepting](f)[above right=of s4]{$f,4$};

        \path[->]
        (s0)edge [above] node [above] {$+2$}
                        node [below] {\text{push }a}(s1)
            edge [] node [below left] {$+0$}(s2)
        (s1)edge [above] node [below right] {$-1$}(s3)
            edge [above] node [above right] {$+1$}(s4)
        (s2)edge [above] node [above] {$+0$}(s4)
            edge [loop left] node [left] {$+1$}(s1)
        (s3)edge [above] node [above right] {$+0$}(f)
            edge [loop left] node [left] {$+1$}(s3)
        (s4)edge [above] node [above left] {$+0$}
                        node [below right] {pop $a$}(f)
        ;
        \end{tikzpicture}
        \caption{An example of a C1PVASS $\mathcal{A}$.}
        \label{fig:C1PVASS}
    \end{figure}
    \cref{fig:C1PVASS} shows a C1PVASS with $6$ states. The second component of the tuple inside the states denotes the lower bound on that state. For instance, $\ell(s_1)=2$. This C1PVASS does not have any run reaching $f$. This is because the only way to make the counter reach $4$ is via $s_2$ or $s_3$. The run through $s_2$ does not push an $a$ into the stack which has to be popped later in order to reach $f$. Also, $s_3$ cannot be reached, since there are only two updates $+2$ and $-1$ before $s_3$ and $\gamma_1\cdot2+\gamma_2\cdot(-1)<2$ for all $\gamma_1,\gamma_2\in(0,1]$.
\end{example}

\subsubsection{Acceptance conditions of a C1PVASS}
There are two classical notions ways of extending (state reachability) acceptance from runs of a PDA to runs $\pi=q_0 \dots q_n$ of C1PVASS, namely: \emph{reachability} and
\emph{coverability} for $k\in\mathbb{R}_{\geq0}$.
\begin{description}
  \item[$k$-Reachability] says the run is accepting if $\restr{\pi}{P_\mathcal{A}}$ is
    accepting in $P_\mathcal{A}$ and
    $c_n = k$.
  \item[$k$-Coverability] says the run is accepting if $\restr{\pi}{P_\mathcal{A}}$ is
    accepting in $P_\mathcal{A}$ and
    $c_n \geq k$.
\end{description}
Like in PDAs, $q_0=(s_0,\alpha_0,c_0)=(s_0,\bot,0)$ means that accepting runs start with the initial configuration. We refer accepting runs according to the above conditions as $k$-reaching and $k$-covering runs, respectively.

We observe that using state-reachability acceptance for the PDA underlying a C1PVASS is no loss of generality.
\begin{lemma}
    The three acceptance conditions: state reachability, empty stack, and both are logspace interreducible even in combination with $k$-reachability and $k$-coverability for C1PVASS.
\end{lemma}

A final simplifying assumption we make is that all the counter updates in the transition function are in the set $\Set{-1,+0,+1}$. This is also no loss of generality due to the following lemma.

\begin{figure}
\centering

        \begin{tikzpicture}[shorten >=1pt,auto,node distance=1.9 cm, scale = 0.7, transform shape]

        \node[state](u)[]{$u$};
        \node[state](us)[below=1cm of u]{};
        \node[state](uv)[right=1.4cm of us]{};
        \node[state](q)[right=of uv]{$q$};
        \node[state](qi)[right=2.5cm of q]{$q_i$};
        \node[](dots1)[below left=1.8cm of qi]{$\vdots$};
        \node[state](qi')[above=.1cm of dots1]{$q_i'$};
        \node[](dots2)[above=2.5cm of dots1]{$\vdots$};
        \node[state](v)[right=1.5cm of qi]{$v$};
        
        \node (X) [draw=red, fit= (q) (qi) (dots1) (dots2), inner sep=0.1cm, fill=red!20, fill opacity=0.2]{};
        \node[text opacity=0.2] [right=1cm of dots1] {\Huge{$\mathcal{G}_+$}};
    
        \path[->,color=blue] 
        (u)   edge [above] node [right] {stack update} (us)

        (us)   edge [above] node [sloped] {push $\#_v$} (uv)
        (uv)  edge [above] node {push bin$(n)$} (q);
        \path[->,color=red] 
        (q)   edge [loop above] node [above] {pop $a_1$}
                                node [below right] {$+1$} (q)
              edge [bend left] node [below] {pop $a_i$} (qi)
              edge [bend left] node [above] {pop $\#_v$} (v)
        (qi)  edge [bend left] node [sloped,anchor=center,below] {push $a_{i-1}$} (qi')
        (qi') edge [bend left] node [sloped,below] {push $a_{i-1}$} (q)
        ;
        \end{tikzpicture}
    \caption{
    Simulating $n$ many $+1$ updates 
    using simple binary arithmetic.}
    \label{fig:bin_gadget}
\end{figure}
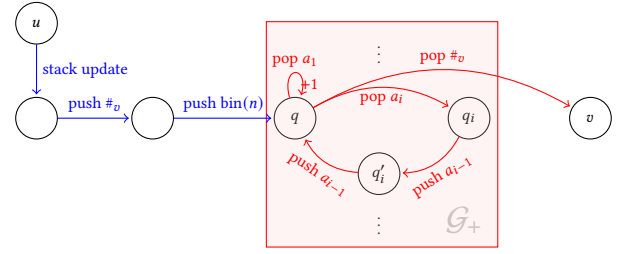

\begin{lemma}\label{lem:3_counter_updates}
    Given a C1PVASS $\mathcal{A}=(S,\Gamma,\delta,\bot,s_0,f,\ell)$, there exists an equivalent C1PVASS\footnote{To be precise: there is a clear relation between their sets of reachable configurations.} with counter updates in the set $\Set{-1,+0,+1}$, which is quadratic in the size of the encoding of $\mathcal{A}$, thus polynomial even if the counter updates are encoded in binary.
\end{lemma}
\begin{proof}[Proof sketch]
    Let $W=\max\Set{|\delta(p,q)_1|: p,q\in S}$ be the largest absolute counter update and $w=\lceil\log_2W\rceil$. The new C1PVASS is $\mathcal{A'}=(S',\Gamma',\delta',\bot,s_0,f,\ell')$, where $\Gamma'$ is obtained by adding new symbols: $a_1,\dots,a_w$ and $\#_s$, for all $s\in S$, to the stack alphabet $\Gamma$. For all new states $s\in S'\backslash S$, $\ell(s)=0$. There are two gadgets, $\mathcal{G}_+$ (depicted as the red box in \cref{fig:bin_gadget}) and $\mathcal{G}_-$ (defined similarly).

    Whenever $\mathcal{A}$ reads a positive update, say, $+n$ on a transition from the state $u$ to $v$, $\mathcal{A'}$ first does the same stack update from $u$ as the transition in $\mathcal{A}$, pushes $\#_v$ (remembering that the next state seen in $S$ must be $v$), and finally pushes the binary encoding of $n$ using the new characters $a_1,\dots,a_w$ to the stack. (For example, the binary encoding of $13$ would be $a_1a_3a_4$.) Inside the gadget, two copies of $a_{i-1}$ are pushed back onto the stack for every copy of $a_i$ being popped, for all $1<i\leq n$. Thus, eventually, there will be exactly $n$ pops of $a_1$, and each time $a_1$ is popped, there is also one $+1$ counter update. Finally, $\#_v$ can be popped and the run enters $v$.
    The gadget $\mathcal{G}_-$ works similarly, except for the fact that it does a $-1$ update on the pop($a_1$) loop on $q$ instead of $+1$.

    Now, we look at the size of $\mathcal{A}'$. Before entering the gadget, at most $w+1$ states are needed to do the stack update, push $\#_v$ and push the binary encoding of $n$. This leads to a quadratic increase in size. Inside the two gadgets, there are two states and three transitions for every $1<i\leq w$, one state $q$, one transition for the loop and $|S|$ transitions exiting the gadget, which leads to a linear increase in size. Hence, the total size of $\mathcal{A'}$ is quadratic in the size of $\mathcal{A}$.
\end{proof}

We also study C1PVASS where all lower-bound guards are $0$.
\begin{definition}[0-C1PVASS]
    A 0-C1PVASS is a C1PVASS $\mathcal{A}=(S,\Gamma,\delta,\bot,s_0,f)$ where $\ell(s)=0$ for all $s \in S$. 
\end{definition}

For 0-C1PVASS, we omit $\ell$.
\textit{Configurations}, \textit{runs}, and \textit{accepting runs} are defined similarly to C1PVASS. Note that, in a configuration $(s,\alpha,c)$, instead of $c\geq\ell(s)$, we now only have the restriction that $c\geq0$, that is, the counter values never go below $0$.

\begin{example}
    In \cref{fig:C1PVASS}, if all the lower bounds were $0$, then we would be able to reach $f$ with a counter value of at least $4$ by taking the run to $s_3$ and taking the self loop a few times before entering $f$ with a counter value of at least $4$. However, the run via $s_2$ would still not be a $4$-reaching run since there is no $a$ to pop from the stack when the run reaches $s_4$.
\end{example}

\subsection{Decision problems}
In the sequel, we focus on the computational complexity of two decision
problems we call reachability and coverability, respectively: Given a C1PVASS and $k \in \mathbb{R}_{\geq 0}$ (in binary),
determine whether it has a $k$-reaching run. Given a C1PVASS and $k \in \mathbb{R}_{\geq 0}$ (in binary), determine
whether it has a $k$-covering run. In addition, we also study the complexity of
the boundedness problem: Given a C1PVASS, determine whether for some $k \in \mathbb{R}_{\geq 0}$ it has no covering run.

\subsection{Our contributions}
We show the following results for 0-C1PVASS.
\begin{enumerate}
    \item For $k>0$, $k$-reachability and $k$-coverability are equivalent.
    \item For $k \geq 0$, $k$-reachability and $k$-coverability are decidable in \textsc{Ptime}.
    \item Boundedness is decidable in \textsc{Ptime}.
    \item If it is bounded, the infimum of all $k \in \mathbb{R}_{\geq 0}$ for which it has no covering run can be computed in polynomial time.
\end{enumerate}
Further, we show that for all $k \geq 0$, $k$-coverability and $k$-reachability are in \textsc{NP} and that boundedness is in \textsc{coNP}, for C1PVASS (with lower-bound guards).

\section{Counter properties of 0-C1PVASS}

We first show a relation between reachability and coverability.

\begin{lemma}\label{lem:cover-reach-equiv}
    The reachability and coverability problems are equivalent for 0-C1PVASS with $k>0$, but $0$-coverability does \textit{not} imply $0$-reachability.
\end{lemma}
\begin{proof}
    By definition, $k$-reachability implies $k$-coverability. To show the converse, take any covering run with counter value at the end of the run being $k+c$, for some $c\geq 0$. Now, we modify the run by scaling all of the counter updates in that run by $\frac{k}{k+c}$. The reader can easily verify that this is indeed a reaching run.

    This proof does not work for $k=0$ since we cannot scale the counter updates by $0$. A simple example for the second part of the lemma would be a 0-C1PVASS with a single transition, which goes from $s_0$ to $f$ with a $+1$ counter update (and no stack update). In this case, $0$ can be covered but not reached.
\end{proof}
From the proof above we directly get the following.
\begin{remark}\label{rem:reach-interval}
    Let $k\in\mathbb{N}_{>0}$. Then, for all $k'\in(0,k]$,
    $k$-reachability implies $k'$-reachability.
\end{remark}
\noindent
We also have the following simple observation about the first nonzero counter update due to our choice of $q_0$.
\begin{remark}\label{rem:first-nonzero-is-pos}
Along any run, the first nonzero counter update must be positive, since the updates cannot be scaled to $0$ and the counter values must always be nonnegative
\end{remark}

Since we have shown that $k$-reachability and coverability are different for $k=0$ but the same for $k>0$, we will first analyse the complexity of $0$-reachability and $0$-coverability in \cref{sec:z-reach-cover}. We then show, in \cref{sec:boundedness}, that boundedness is decidable in \textsc{PTime} and that, if a C1PVASS is bounded, computing the infimum upper bound is also in \textsc{PTime}. Finally, in \cref{sec:pos-cover-reach}, we leverage our algorithm for boundedness to show that $k$-reachability and $k$-coverability for $k>0$ are also in \textsc{PTime}.

\subsection{0-reachability and 0-coverability}\label{sec:z-reach-cover}

In this section, we show that both $0$-reachability and $0$-coverability are in \textsc{PTime} for 0-C1PVASS by reducing the problems to checking nonemptiness of PDAs (with an empty alphabet).

\begin{theorem}\label{thm:zreach-zcover-ptime}
    The $0$-reachability and $0$-coverability problems for a 0-C1PVASS are decidable in \textsc{PTime}.
\end{theorem}
The result follows from the fact that checking nonemptiness of the language of a PDA can be done in polynomial time (see, e.g.~\cite[Proof of Lemma 4.1]{Hopcroft+Ullman/79/Introduction}) and the following lemma.

\begin{lemma}\label{lem:covreach0}
    The $0$-reachability and $0$-coverability problems for 0-C1PVASS are polynomial time reducible to the nonemptiness of the language of a PDA.
\end{lemma}
\begin{proof}
    Along any run of the 0-C1PVASS, the first counter update which is not zero must be positive, and after this positive update, all the upcoming negative updates can be scaled down sufficiently so that the counter never goes below some $0<\varepsilon<1$ small enough (cf. \cite[Proposition 14]{coca}). The above two observations give us the following reduction.

    Let $\mathcal{A}$ be the 0-C1PVASS. Make two copies of $\mathcal{A}$ without the counter: namely $\mathcal{P}_0$ and $\mathcal{P}_1$ are copies of the PDA underlying $\mathcal{A}$. Remove from $\mathcal{P}_0$ all the transitions that were not a $+0$ counter update in $\mathcal{A}$. For each transition in $\mathcal{A}$ with a $+1$ counter update, add the transition from the corresponding source state in $\mathcal{P}_0$ to the corresponding target state in $\mathcal{P}_1$. 
    Call the resulting PDA $\mathcal{P}$ (see \cref{fig:cover0_ptime} for a graphical depiction). The accepting states of $\mathcal{P}$ are the copies of the accepting states of $\mathcal{A}$ in both $\mathcal{P}_0$ and $\mathcal{P}_1$.

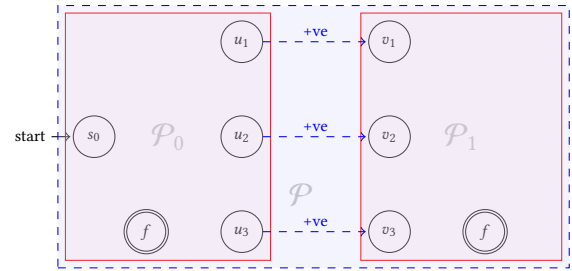
\begin{figure}[h]
    \centering
    \begin{tikzpicture}[shorten >=1pt,auto,node distance=1.9 cm, scale = 0.7, transform shape]

    \node[state,initial](s0)[]{$s_0$};
    \node[state](u2)[right=2cm of s0]{$u_2$};
    \node[state](u1)[above=1cm of u2]{$u_1$};
    \node[state](u3)[below=1cm of u2]{$u_3$};
    \node[state,accepting](f)[left=1cm of u3]{$f$};
    \node (X) [draw=red, fit= (s0) (u1) (u3), inner sep=0.1cm, fill=red!20, fill opacity=0.2]{};
    \node[text opacity=0.2] at (X.center) {\Huge{$\mathcal{P}_0$}};
    
    \node[state](v2)[right=2cm of u2]{$v_2$};
    \node[state](v1)[above=1cm of v2]{$v_1$};
    \node[state](v3)[below=1cm of v2]{$v_3$};
    \node[](invis)[right=2.5cm of v2]{};
    \node[state,accepting](f')[right=1cm of v3]{$f$};

    \node (X1) [draw=red, fit= (v1) (v3) (invis), inner sep=0.1cm, fill=red!20, fill opacity=0.2]{};
    \node[text opacity=0.2] at (X1.center) {\Huge{$\mathcal{P}_1$}};
    
    \node (Y) [draw=blue, fit= (s0) (u1) (v3) (invis), inner sep=0.2cm, fill=blue!20, fill opacity=0.2, dashed] {};
    \node[text opacity=0.2,below right=0.7cm of u2] {\Huge{$\mathcal{P}$}};
    
    \path[->,color=blue,dashed] 
    (u1)   edge [above] node [align=center] {+ve} (v1)
    (u2)   edge [above] node [align=center] {+ve} (v2)
    (u3)   edge [above] node [align=center] {+ve} (v3)
    ;
    \end{tikzpicture}
    \caption{The construction of the PDA $\mathcal{P}$ where $\mathcal{P}_0$ is a copy of $\mathcal{A}$ obtained by removing all the counter updates and removing all the transitions that have a nonzero counter update; the transitions from $\mathcal{P}_0$ to $\mathcal{P}_1$ are exactly the transitions in $\mathcal{A}$ with a positive counter update.}
    \label{fig:cover0_ptime}
\end{figure}
    \begin{claim}
        There exists a 0-covering run of $\mathcal{A}$ iff $L(\mathcal{P})$ is nonempty.
    \end{claim}
    \begin{proof}
        By construction, if $L(\mathcal{P})$ is empty then there is no $0$-covering run in $\mathcal{A}$ where the first nonzero counter update is positive and there is no $0$-covering run in $\mathcal{A}$ where all updates are $+0$. It follows that $\mathcal{A}$ has no covering run.

        If there does not exist a 0-covering run in $\mathcal{A}$, then there is no run starting in $s_0$ and ending in $f$ in the underlying $PDA$ of $\mathcal{A}$ with the property that the first nonzero counter update, if any, is positive. Thus, by construction of $\mathcal{P}$, there is no accepting run ending in the copy of $f$ in $\mathcal{P}_0$ since the only transitions are ones which had $+0$ updates in $\mathcal{A}$. There is also no accepting run ending in the copy of $f$ in $\mathcal{P}_1$ since these runs correspond to runs in $\mathcal{A}$ with $+1$ being the first nonzero update and with at least one $+1$ update.  
    \end{proof}

This shows that $0$-coverability in $\mathcal{A}$ is equivalent to asking whether $L(\mathcal{P})$ is nonempty. For reachability, we note that the following are sufficient and necessary.
\begin{enumerate}
    \item There exists a run starting from the initial configuration and ending at $f$ with only $+0$ counter updates, or
    \item there exists a run starting from the initial configuration and ending at $f$ where the first nonzero update is positive and the last nonzero update is negative.
\end{enumerate}

Observe that the PDA $\mathcal{P}$ already takes care of the accepting runs with only zero updates (i.e. those that never leave $\mathcal{P}_0$) and the property of having a positive number as the first nonzero counter update is satisfied by all runs that reach $\mathcal{P}_1$. We now modify $\mathcal{A}$ by adding a copy $\mathcal{P}_0'$ of $\mathcal{P}_0$ to the right of $\mathcal{P}_1$, with the only transitions from $\mathcal{P}_1$ to $\mathcal{P}_0'$ being ones with a negative counter update in the C1PVASS $\mathcal{A}$. Finally, the accepting states are all copies of accepting states from $\mathcal{A}$ in $\mathcal{P}_0$ or $\mathcal{P}'_0$ (the new PDA $\mathcal{P}$ is depicted in \cref{fig:reach0_ptime}).

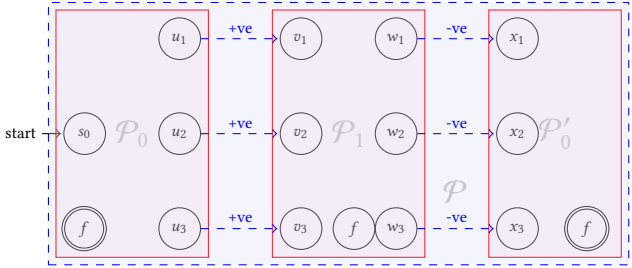
\begin{figure}[h]
    \centering
    \begin{tikzpicture}[shorten >=1pt,auto,node distance=1 cm, scale = 0.7, transform shape]

    \node[state,initial](s0)[]{$s_0$};
    \node[state](u2)[right=1cm of s0]{$u_2$};
    \node[state](u1)[above=1cm of u2]{$u_1$};
    \node[state](u3)[below=of u2]{$u_3$};
    \node[state,accepting](f)[left=1cm of u3]{$f$};
    \node (X) [draw=red, fit= (s0) (u1) (u3), inner sep=0.1cm, fill=red!20, fill opacity=0.2]{};
    \node[text opacity=0.2] at (X.center) {\Huge{$\mathcal{P}_0$}};
    
    \node[state](v2)[right=1.5cm of u2]{$v_2$};
    \node[state](v1)[above=1cm of v2]{$v_1$};
    \node[state](v3)[below=of v2]{$v_3$};
    \node[state](f")[right=.2cm of v3]{$f$};
    \node[state](w2)[right=1cm of v2]{$w_2$};
    \node[state](w1)[above=1cm of w2]{$w_1$};
    \node[state](w3)[below=of w2]{$w_3$};
    \node (X1) [draw=red, fit= (v1) (v3) (w2), inner sep=0.1cm, fill=red!20, fill opacity=0.2]{};
    \node[text opacity=0.2] at (X1.center) {\Huge{$\mathcal{P}_1$}};

    \node[state](x2)[right=1.5cm of w2]{$x_2$};
    \node[state](x1)[above=1cm of x2]{$x_1$};
    \node[state](x3)[below=of x2]{$x_3$};
    \node[state,accepting](f')[right=.5cm of x3]{$f$};
    \node[](invis)[right=1.2cm of x2]{};
    \node (X2) [draw=red, fit= (x1) (x3) (invis), inner sep=0.1cm, fill=red!20, fill opacity=0.2]{};
    \node[text opacity=0.2] at (X2.center) {\Huge{$\mathcal{P}_0'$}};
    
    \node (Y) [draw=blue, fit= (s0) (u1) (v3) (invis), inner sep=0.2cm, fill=blue!20, fill opacity=0.2, dashed] {};
    \node[text opacity=0.2]()[below right=0.7cm of w2] {\Huge{$\mathcal{P}$}};
    
    \path[->,color=blue,dashed] 
    (u1)   edge [above] node [align=center] {+ve} (v1)
    (u2)   edge [above] node [align=center] {+ve} (v2)
    (u3)   edge [above] node [align=center] {+ve} (v3)
    (w1)   edge [above] node [align=center] {-ve} (x1)
    (w2)   edge [above] node [align=center] {-ve} (x2)
    (w3)   edge [above] node [align=center] {-ve} (x3)
    ;
    \end{tikzpicture}
    \caption{The construction of the PDA $\mathcal{P}$ where $\mathcal{P}_0$ and $\mathcal{P}_0'$ are copies of $\mathcal{A}$ obtained by the counter and removing all the transitions that have a nonzero counter update; the transitions from $\mathcal{P}_0$ to $\mathcal{P}_1$ are exactly the transitions in $\mathcal{A}$ with a positive counter update and those from $\mathcal{P}_1$ to $\mathcal{P}_0'$ are exactly the ones with a negative counter update.}
    \label{fig:reach0_ptime}
\end{figure}

\begin{claim}
    There exists a $0$-reaching run of $\mathcal{A}$ iff $L(\mathcal{P})$ is nonempty.
\end{claim}

The proof is similar to the one given for $0$-coverability and follows from the fact that any reaching run with final counter value $0$ satisfies exactly one of the sufficient and necessary conditions stated earlier. This concludes the proof of \cref{lem:covreach0}.
\end{proof}

\subsection{Boundedness for 0-C1PVASS}\label{sec:boundedness}

In this section, we first analyze the complexity of deciding whether a 0-C1PVASS is bounded or not. If it is bounded, we provide a bound which is polynomial (when encoded in binary) in the size of the encoding of the C1PVASS. We next show that, for a bounded 0-C1PVASS, the ``tight'' bound, that is, 
\begin{equation}\label{eqn:bound}
b=\inf\Set{k\in\mathbb{R}\mid \mathcal{A}\text{ has no $k$-covering run}}
\end{equation}
is an integer and the natural decision problem associated to finding $b$ is in $\textsc{PTime}$.
First, we convert the 0-C1PVASS into a PDA as we did in the proof of \cref{lem:covreach0} in \cref{fig:cover0_ptime}, further modify its alphabet, and observe some properties about the resulting PDA.

Let $\mathcal{A}$ be the 0-C1PVASS. Make two copies $\mathcal{\mathcal{P}}_0$ and $\mathcal{P}_1$ of $\mathcal{A}$ without the counter. Next, remove from $\mathcal{P}_0$ all the transitions that were not a $+0$ counter update in $\mathcal{A}$ and add, for each transition in $\mathcal{A}$ with a positive counter update, a transition from $\mathcal{P}_0$ to $\mathcal{P}_1$. The copies of accepting states in $\mathcal{P}_0$ and $\mathcal{P}_1$ are all accepting in the resulting PDA, which we call $\mathcal{P}$ (see \cref{fig:cover0_ptime}). To obtain $\mathcal{P}'$ from $\mathcal{P}$, we modify its alphabet. The alphabet $\Sigma$ of $\mathcal{P}'$ is unary, i.e. $\Sigma = \{a\}$. The transitions of $\mathcal{P}'$ read $a$ if they had a $+1$ counter update in $\mathcal{A}$ and read the empty word $\epsilon$ otherwise.

One can see a bijection between accepting runs in $\mathcal{A}$ and $\mathcal{P'}$. Let $\pi$ be an accepting run in $\mathcal{P'}$. The corresponding run in $\mathcal{A}$ has the property that the first nonzero update is a positive update (i.e., a transition from $\mathcal{P}_0$ to $\mathcal{P}_1$) which makes it an accepting run in $\mathcal{A}$. Similarly, an accepting run in $\mathcal{A}$ must have the first nonzero update a $+1$, hence, it is also an accepting run in $\mathcal{P'}$ by construction.

The lemma below follows immediately from the construction.

\begin{lemma}\label{lem:bijection-CPVASS-PDA-zboundedness}
    $a^m\in L(\mathcal{P}')$ iff there is an accepting run in $\mathcal{A}$ with exactly $m$ many $+1$ updates.
\end{lemma}

For all $0<\varepsilon<1$, and an accepting run in the PDA $\mathcal{P'}$, in the corresponding run in $\mathcal{A}$, one can choose $\gamma=1$ for all the $+1$ updates and $\gamma\in(0,1]$ small enough, for all negative updates, so that their sum is in the interval $(-\varepsilon,0)$. This leads to the following result.

\begin{lemma}\label{lem:A-Pprime-bounded}
    The cardinality of $L(\mathcal{P}')$ is bounded iff the 0-C1PVASS $\mathcal{A}$ is bounded. Moreover, if the maximum length of a word accepted by $\mathcal{P}'$ is $p \in \mathbb{N}$ then $b=p$, where $b$ is as in \Cref{eqn:bound}.
\end{lemma}

\noindent
There are \textsc{Ptime} algorithms (see, e.g., \cite[Theorem 6.6]{Hopcroft+Ullman/79/Introduction}) to determine whether the language of a PDA is finite. We thus get:

\begin{theorem}\label{thm:zboundedness-ptime}
    Deciding boundedness of a 0-C1PVASS $\mathcal{A}$ is in \textsc{PTime}. Moreover, the bound can be at most $2^{(4|\mathcal{A}|)^9}$, where $|\mathcal{A}|$ is the size of the encoding of $\mathcal{A}$.
\end{theorem}
\begin{proof}
The first part of the proof follows from the previous discussion.
The second part is equivalent to bounding the length of the longest world accepted by $\mathcal{P'}$ by $2^{(4|\mathcal{A}|)^9}$. Recall that CFGs are expressively equivalent to PDAs and that CNF is a normal form for CFGs where production rules are of the form $A\rightarrow BC$, $A\rightarrow a$ or $A\rightarrow\varepsilon$. 
As before, if the bound on the length of words in $\mathcal{P}'$ exists, say $p$, it will also be a bound on the largest counter reachable in $\mathcal{A}$. 
The size of $\mathcal{P}'$ is at most $2|\mathcal{A}|$.
Now, we translate $\mathcal{P'}$ to an equivalent CFG in CNF. Towards this, converting a PDA to one without empty stack updates can at most double the size, that is, it can have size $4|\mathcal{A}|$. Converting such a PDA to a CFG and then to CNF (without useless rules) \cite[Theorem 4.5]{Hopcroft+Ullman/79/Introduction} leads to at most a cubic expansion in each step. That is, the CNF will have size at most $k=(4|\mathcal{A}|)^9$. This CFG in CNF has a finite language iff there is no loop \cite[Theorem 6.6]{Hopcroft+Ullman/79/Introduction}: all variables cannot derive a string containing the same variable. Thus, because of the form of production rules of CFGs in CNF, which can at most double the length of the final word at every production, we get a bound of $2^k$ on the longest word in the language of the CFG.
\end{proof}

Using this upper bound on the largest reachable counter for a bounded 0-C1PVASS, we argue the tight upper bound is an integer and give an algorithm to compute it.

\begin{remark}
Using \cref{lem:bijection-CPVASS-PDA-zboundedness} and the fact that 
nonnegative updates can be scaled down arbitrarily,
one can see that the bound $b$ defined in \Cref{eqn:bound} is a nonnegative integer when it exists.
\end{remark}

\begin{theorem}\label{thm:tight_bound}
    The tight upper bound of a bounded 0-C1PVASS can be computed in \textsc{PTime}.
\end{theorem}

\begin{proof}
    The idea for the proof comes from \cite{stackoverflow-min-len-word-PDA} which gives a \textsc{PTime} algorithm to find the shortest word accepted by a CFG.
    
    Assume the language is not empty. Construct the PDA described in \cref{lem:A-Pprime-bounded}. We know that if $m$ is the length of a longest word accepted by the PDA $\mathcal{P}$, then $m$ is the tight bound. We also know, by \cref{thm:zboundedness-ptime}, that $m\leq 2^k$ where $k=(4|\mathcal{A}|)^9$. We construct a grammar $(V,\Sigma,P,S)$ in CNF for the PDA. This grammar has size at most $2^{(4|\mathcal{A}|)^9}$, as shown in the previous theorem. Since the grammar is in CNF, all productions are of the form $A\rightarrow BC$ or $A\rightarrow a$ and the language of all variables is nonempty. 
    
    Define the function $N:V\rightarrow\mathbb{N}$ such that $N(A)$ is the length of the longest word produced by the variable $A$, for all $A\in V$. The following algorithm computes $N(A)$ for all $A\in V$.
    \begin{enumerate}
        \item Initialize $W(A)=0$ for all $A\in V$, $W(a)=1$ for all $a \in T$.
        \item Repeat, for all $A$ and all productions with head $A$:\[
            W(A)=\begin{cases}
                \max\{W(B)+W(C),W(A)\} & \text{if }A\rightarrow BC;\\
                \max\{W(a),W(A)\} & \text{if }A\rightarrow a.
            \end{cases}
        \] until we reach a fix point (we know a fix point will be reached eventually since the length of words is bounded).
        \item Output the vector $W(V)$.
    \end{enumerate}
    We know that the above algorithm terminates since the length of the longest word is bounded. It remains to show that it terminates in polynomially many iterations. Each iteration has $|V||P|$ comparisons of numbers bounded by $2^{(4|\mathcal{A}|)^9}$, and we know that such numbers can be compared in time polynomial in $|\mathcal{A}|$. Hence, showing that the fix point is obtained in polynomially many iterations of the algorithm suffices to establish that the tight bound can be obtained in polynomial time.

    Consider the directed graph with $V \cup T$ as vertices and where we add the edge $(A,\beta)$, where $A\in V$ and $\beta\in V\cup T$, if and only if there is a production in $P$ whose head is $A$ and with $\beta$ in its body. The graph can be shown to be acyclic, since the language of the grammar is finite and the language of every variable is nonempty. Now, every iteration of the algorithm induces a labelling of the vertices of the graph via $W$. Observe that the label of a vertex only changes if the label of one of its immediate successors changes. It follows that the fix point is reached after at most $|V|$ iterations.    
\end{proof}

However, we do not know if the bound itself is reachable. That is, the interval of nonzero reachable values can be $(0,b)$ or $(0,b]$.

\begin{lemma}\label{lem:closed-open-bound}
    The set of all reachable values in a 0-C1PVASS is closed on the right (i.e., the bound $b$ can be reached) iff there is an accepting run for $a^b$ in $\mathcal{P'}$ which does not contain any $-1$ transitions from $\mathcal{A}$. 
\end{lemma}
The proof follows from the simple fact that any $-1$ update in $\mathcal{A}$ cannot be scaled down to $0$, and $b$ is an upper  bound on the counter value in the final configuration of any accepting run.

\cref{lem:closed-open-bound} gives us an easy way to check whether the interval of all reachable counter values is closed on the right. Remove all transitions from $\mathcal{P'}$ which correspond to a $-1$ update transition in $\mathcal{A}$. This PDA $\mathcal{P''}$ will accept $a^m$ iff $\mathcal{A}$ has an accepting run with exactly $m$ many $+1$ updates and no negative updates.

\subsection{$k$-reachability and $k$-coverability for $k>0$}\label{sec:pos-cover-reach}

The following stronger theorem implies that both $k$-reachability and coverability are in \textsc{PTime} for all $k\geq0$.

\begin{theorem}\label{thm:exact-zreach-interval}
     The interval of all reachable counter values of a 0-C1PVASS is computable in polynomial time.
\end{theorem}
\begin{proof}
    Use \cref{thm:zreach-zcover-ptime} to decide whether $0$ is reachable. If so, the interval is closed on the left, open otherwise. Next, use \cref{thm:zboundedness-ptime} to decide if the highest reachable counter value is bounded. If not, the upper bound will be $\infty$ (and thus open).  If it is bounded, use \cref{thm:tight_bound} to compute the tight bound $b$. Finally, use \cref{lem:closed-open-bound} to find whether the interval is closed on the right.
\end{proof}

Note that \textsc{Ptime}-hardness for all the problems in this section follows from the nonemptiness problem for PDAs being \textsc{Ptime}-hard  (see, e.g. \cite[Prop. 1]{dmitry-rupak2014}). For coverability and reachability this is immediate, for boundedness one can add a self loop with a positive counter update on accepting states.

\section{Counter properties of C1PVASS}
In this section, we show that coverability and reachability for C1PVASS are decidable in \textsc{NP}. Finally, we comment on how our treatment of boundedness from the previous section adapts almost identically to C1PVASS to yield, in this case, a complexity of \textsc{coNP}.  

Both for coverability and reachability, we proceed as follows. First, we convert the given C1PVASS into a PDA $\mathcal{P}$ such that the \emph{Parikh image} of a word accepted by the PDA satisfies some quantifier-free Presburger formula $\varphi$ if and only if the C1PVASS has an accepting run. Then, we use a construction from~\cite[Theorem 4]{Seidl} (later corrected in~\cite{HagueLinFix}) to obtain, in polynomial time, an existential Presburger formula $\varphi_L$ whose models correspond to the Parikh images of words in the language of $\mathcal{P}$. The problem thus reduces to checking satisfiability of the existential Presburger formula $\varphi \land \varphi_L$. The result follows since satisfiability for such formulas is known to be \textsc{NP}-complete~\cite{HaaseSIGLOG}.

\begin{theorem}
    $k$-coverability and $k$-reachability for C1PVASS are decidable in \textsc{NP}.
\end{theorem}

\subsection{$k$-coverability for C1PVASS}\label{sec:cover-C1PVASS}
Note that, for C1PVASS, $k$-coverability is equivalent to \emph{state reachability}, i.e. without asking for the final counter value to be at least some given value: to check $k$-coverability, we add a new final state with lower-bound guard $k$ and transitions from the old final state(s) to this new state with $+0$ counter updates and no stack update. Because of this, we will focus on state reachability as acceptance condition and omit $k$ when speaking of coverability in the sequel.

Let $\mathcal{A}=(S,\Gamma,\delta,\bot,s_0,f,\ell)$ be the C1PVASS. Recall that $\ell:S\rightarrow\mathbb{N}$ is the mapping from states to the lower bounds on those states. That is, $\ell(s)=x$ implies that the counter value must be at least $x$ in order to enter the state $s$. Let $n=\vert S\vert$ be the number of states. We have the assumption, from \cref{lem:3_counter_updates}, that the only counter updates in the C1PVASS are in the set $\Set{-1,+0,+1}$. Let $m+1\leq n$ be the size of the range of $\ell$. That is, there are $m+1$ distinct lower bounds $0=\ell_0<\ell_1<\ell_2<\dots<\ell_m$ that occur in the C1PVASS $\mathcal{A}$. Note that $0$ must be one of the lower bounds since $\ell(s_0)=0$ in order for any run to exist.

Now, we construct the PDA followed by the Presburger formula.
The PDA $\mathcal{P}$ has $4(m+1)$ ``blocks'' and its alphabet is $\Sigma = \{a_i, a'_i, b_i \mid 0 \leq i \leq m + 1\}$. Each block is a subPDA (so, we ignore counter updates) of the C1PVASS with some restrictions. For each $0 \leq i \leq m$, the 4 types of blocks we use all have copies of the same set of states: all $s \in S$ such that $\ell(s) \leq \ell_i$. 
\begin{enumerate}
    \item \cgreencolorA{i}{} The transitions come from those in $\mathcal{A}$ with counter update $+0$ and they all read $\epsilon$ in $\mathcal{P}$;
    \item \cgreencolorA{i}{+} Here, from the $+0$ and $+1$ transitions in $\mathcal{A}$ and the PDA $\mathcal{P}$ reads an $a_i$ on the $+1$ transitions;
    \item \cbluecolorA{i}{-} Here, from the $+0$ and $-1$ transitions in $\mathcal{A}$ and the PDA $\mathcal{P}$ reads an $b_{i}$ on the $-1$ transitions;
    \item \credcolorA{i}{\pm} And here, from the $-1$, $+0$ and $+1$ transitions in $\mathcal{A}$ and the PDA will read $b_i$ on $-1$ and $a_i$ on the $+1$ transitions.
\end{enumerate}

\begin{figure*}[h]
    \centering
    \begin{tikzpicture}[shorten >=1pt,auto,node distance=3.5cm, scale = 0.7, transform shape]
        \tikzstyle{greenblock} = [shape=rectangle, draw=green, minimum size=2cm,fill=green!20]
        \tikzstyle{blueblock} = [shape=rectangle, draw=blue, minimum size=2cm,fill=blue!20]
        \tikzstyle{redblock} = [shape=rectangle, draw=red, minimum size=2cm,fill=red!20]


        \node[greenblock,accepting,label={[label distance=.1cm]330:\cgreencolorA{i}{}}](p0)[]{$+0:\epsilon$};
        \node[greenblock,accepting,label={[label distance=.1cm]210:\cgreencolorA{i}{+}}](p1)[right=2cm of p0]{$+1:a_i$};
        \node[greenblock](p2)[right=of p1,label={[label distance=.1cm]330:\cgreencolorA{i+1}{}}]{};
        \node[blueblock,accepting](d0)at($(p0.300)+(300:2cm)$)[label={[label distance=.1cm]330:\cbluecolorA{i-1}{-}}]{$-1:b_i$};
        \node[redblock,accepting](d1)at($(d0.300)+(300:2cm)$)[label={[label distance=.1cm]210:\credcolorA{i}{\pm}}]{$-1:b_i,+1:a_i$};
        \node[redblock](d2)[right=of d1,label={[label distance=.1cm]330:\credcolorA{i+1}{\pm}}]{};        


        \node (Y) [draw=black, fit= (p0) (p1) (d1), inner sep=0.2cm, fill=yellow!20, fill opacity=0.2, dashed] {};
        
        \path[->]
            (p0) edge node [above,align=center] {$+1$} node [below,align=center] {$a_i$} (p1)
            (p1) edge node [above,align=center] {$+1$} node [below,align=center] {$a_i'$} (p2)
            (p0) edge [bend left] node [above,align=center] {$+1$} node [below,align=center] {$a_i'$} (p2)
            (p0) edge node [right,align=center] {$-1$} node [left,align=center] {$b_i$} (d0)
            (d0) edge node [right,align=center] {$+1$} node [left,align=center] {$a_i$} (d1)
            (p1) edge [bend left] node [left,align=center] {$-1$} node [right,align=center] {$b_i$} (d1)
            (d1) edge node [above,align=center] {$+1$} node [below,align=center] {$a_i'$} (d2)
        ;
        \end{tikzpicture}
        \caption{A slice of the PDA $\mathcal{P}$ constructed for $k$-coverability of a C1PVASS. The subscript being $i$ for $0\leq i\leq m$ of a block (for example, $i$ in \cgreencolorA{i}{}) denotes that all the states in the block have lower bounds at most $\ell_i$. Note $+0 : \epsilon$ is omitted unless it is the only option for transitions in the block.}
    \label{fig:C1PVASS-lb-to-PDA}
\end{figure*}
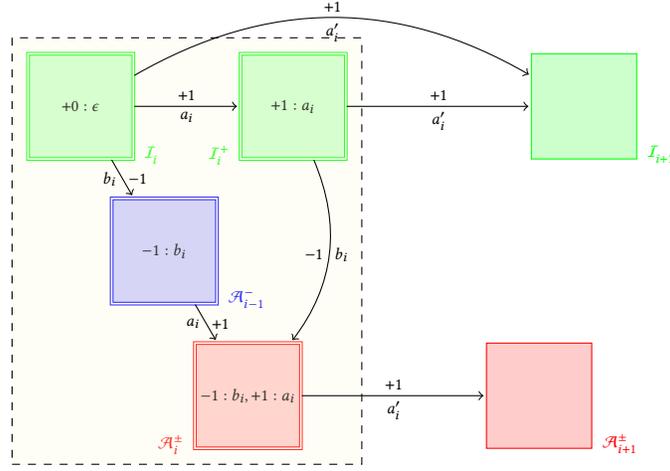

\noindent
\cref{fig:C1PVASS-lb-to-PDA} depicts how the blocks are connected in  what we will henceforth call a \emph{slice}, i.e. $\cgreencolorA{i}{}$, $\cgreencolorA{i}{+}$, $\cbluecolorA{i-1}{-}$, and $\credcolorA{i}{\pm}$, for some $0 \leq i \leq m$. It also shows how the slices themselves are connected.
Note that the transitions in the PDA do not actually have the counter updates $-1,+0,+1$, but we include them in the explanation and figure for clarity.  The accepting states of $\mathcal{P}$ are all the copies of accepting states in $\mathcal{A}$.

Now, we define a Presburger formula for the Parikh images of accepting runs in $\mathcal{P}$ that correspond to the accepting runs in $\mathcal{A}$ ($\#_a$ denotes the number of $a$'s read during the run).
\begin{equation}\label{eqn:presb}
        \begin{aligned}
            \bigwedge_{k=1}^m\left(
            \begin{aligned}
                \big(\#_{a_{k-1}'}&=0\big)\vee\\
                \bigg(\Big(\sum_{i=0}^{k-1}\#_{a_i}+\#_{a_i'}\geq \ell_k\Big)&\wedge\Big(\sum_{i=0}^{k-1}\#_{b_i}=0\Big)\bigg)\vee\\
                \bigg(\Big(\sum_{i=0}^{k-1}\#_{a_i}+\#_{a_i'}>\ell_k\Big)&\wedge\Big(\sum_{i=0}^{k-1}\#_{b_i}>0\Big)\bigg)
            \end{aligned}
            \right)
        \end{aligned}
    \end{equation}


\begin{theorem}\label{thm:C1PVASS+lb_to_PDA+Pres}
    For all runs $\pi$ in $\mathcal{A}$, there is one in $\mathcal{P}$ with the same sequence of states whose Parikh image satisfies the Presburger formula from \cref{eqn:presb} if and only if $\pi$ is accepting in $\mathcal{A}$.
\end{theorem}

\noindent
Before proving this, we give some auxiliary lemmas.

\begin{lemma}\label{lem:extremal_runs}
    Let $0<\varepsilon<1$. For any run $\pi$ in $\mathcal{A}$, there is another run $\pi'$ with the same sequence of states such that all the $+1$ counter updates in the run are scaled up to $1$ and all the $-1$ updates in the run are scaled down so as to add up to $-\varepsilon$.
\end{lemma}

\noindent
The proof follows from a few simple observations: Since $\pi$ is a run and $\pi'$ is a run with the same sequence of states, the stack will behave the same in both runs. For the counter, since we only have lower bounds, and we chose small coefficients for negative updates, all the counter updates in $\pi'$ are greater than, or equal to the counter updates in $\pi$, hence satisfying all lower bounds along the run.
For the remainder of the section, we therefore assume that all runs of $\mathcal{A}$ have $+1$s scaled up to $1$ and $-1$s scaled down so that they add up to $-\varepsilon > -1$.

\begin{lemma}\label{lem:int_nonint_states}
    A run ends in a green state (i.e., a state in \cgreencolorA{i}{} or \cgreencolorA{i}{+} for some $0\leq i\leq m$) in $\mathcal{P}$ if and only if there was no $b_j$ read along the run for all $0\leq j\leq m$. 
\end{lemma}

\noindent
This follows from the construction since there is no path from the blue states (states in \cbluecolorA{i}{-} for $0\leq i\leq m$) or the red states (states in \credcolorA{i}{\pm} for $0\leq i\leq m$) to any green state.
Intuitively, the counter values are integers when reaching green states and nonintegers when reaching blue or red states.

\subsubsection{Proof of \cref{thm:C1PVASS+lb_to_PDA+Pres}}

\begin{proof}[Proof ($\impliedby$)]
    We first prove the easy part, that is showing that an accepting run that satisfies the Presburger formula in $\mathcal{P}$ corresponds to an accepting run in $\mathcal{A}$ with the same sequence of states.

    Let $\pi$ be the run. We show that $\pi$ is also an accepting run in $\mathcal{A}$ with all the $+1$ counter updates scaled up to $1$ and the $-1$ updates scaled down to add up to $-\varepsilon > -1$. Since all the transitions in $\mathcal{P}$ are also transitions in $\mathcal{A}$, we just need to show that the lower bounds along the states in $\pi$ from $\mathcal{A}$ are satisfied. Let $s$ be a state in $\pi$.
    \begin{enumerate}
        \item \label{enum:1} If $s\in\cgreencolorA{i}{}$ or $s\in\cgreencolorA{i}{+}$, we know from \cref{lem:int_nonint_states} that there was no $b_j$ seen before for $0\leq j\leq m$. Since the Parikh image of $\pi$ in $\mathcal{P}$ satisfies \cref{eqn:presb}, we know that the second disjunct in the braces must hold for $k=i$. This implies that there were at least $\ell_i$ many $a$'s seen before reaching $s$. Now, we see that in the run $\pi$ in $\mathcal{A}$, before reaching $s$, there were at least $\ell_i$ many $+1$ updates (since each $a_j$ in $\mathcal{P}$ corresponds to a $+1$ update in $\mathcal{A}$) and no $-1$ updates (since each $b_j$ in $\mathcal{P}$ corresponds to a $-1$ update in $\mathcal{A}$). Moreover, the lower bound at $s$ can be at most $\ell_i$ by the definition of the blocks. Hence, the lower bound on $s$ is satisfied.
        \item \label{enum:2} If $s\in\cbluecolorA{i-1}{-}$, it means that there was at least one $b_i$ seen before reaching $s$ since we know that the only way to enter the block \cbluecolorA{i-1}{-} is from \cgreencolorA{i}{}. We know that there are at least $\ell_i$ many $a$'s seen before entering \cgreencolorA{i}{} (from \cref{enum:1}). Note that after entering \cgreencolorA{i}{} until reaching $s$, there are only $+0$ and $-1$ update transitions, of which, the $-1$ updates can be scaled down to $-\varepsilon$. Hence, the counter value at $s$ will be in the interval $(\ell_i-1,\ell_i)$. The lower bound on $s$ can be at most $\ell_i-1\geq\ell_{i-1}$. Hence, the lower bound is satisfied.
        \item \label{enum:3} If $s\in\credcolorA{i}{\pm}$, we consider the 3 possible ways to enter \credcolorA{i}{\pm}:
        \begin{itemize}
            \item From \credcolorA{i-1}{\pm}: In this case, we know that the run, which started in a green state, must have seen a $b_j$ for some $0\leq j<i$ before entering \credcolorA{i-1}{\pm}, hence, the third disjunct in the braces of \cref{eqn:presb} must hold for $k=i$. Hence, there must be at least $\ell_i+1$ many $+1$ updates seen before entering \credcolorA{i}{\pm}. Since all the $-1$ updates can be scaled down to $-\varepsilon$, the counter value for all the states in \credcolorA{i}{\pm} will be in the interval $(\ell_i,\ell_i+1)$. We know that the lower bounds of all the states in \credcolorA{i}{\pm} are at most $\ell_i$, hence, the lower bound of $s$ is satisfied.
            \item From \cgreencolorA{i}{+}: In this case, we see that the only way to enter \cgreencolorA{i}{+} is by a $+1$ update from \cgreencolorA{i}{}, and we know from \cref{enum:1} that the counter values for such runs are at least $\ell_i$, which means that the counter values when the run is in \cgreencolorA{i}{+} is at least $\ell_i+1$ since there are no negative updates in the green states. We can now repeat the argument in the previous bullet to argue that the lower bound in $s$ is satisfied.
            \item From \cbluecolorA{i-1}{-}: In this case, we know from \cref{enum:2} that there were at least $\ell_i$ many $+1$ updates seen before entering \cbluecolorA{i-1}{-}. Now, since we need a $+1$ update to go from \cbluecolorA{i-1}{-} to \credcolorA{i}{\pm}, we can repeat the same argument as before to conclude the lower bound at $s$ is satisfied.\qedhere
        \end{itemize}
    \end{enumerate} 
\end{proof}

\begin{proof}[Proof ($\implies$)]
    Let $\pi=q_0\dots q_n$ be an accepting run in $\mathcal{A}$, where all the $+1$ updates are scaled up to $1$ and all the $-1$ updates are scaled down to $-\delta$ so that they all add up to $-\varepsilon > -1$. Note that the counter values, $c_i$ in $q_i = (s_i,\alpha_i,c_i)$, along the run will either be integers or they will be $p-k\delta$ for some $p,k\in\mathbb{N}$ where $0<k\delta \leq \varepsilon < 1$.
    We have to show that there is a run in $\mathcal{P}$ with the same sequence of states such that the Parikh image of this run satisfies the Presburger formula. Technically, we will prove the following.
    \begin{claim}
        There is an accepting run $\pi' = (s'_0,w_0,\alpha_0)\dots(s'_n,w_n,\alpha_n)$ in $\mathcal{P}$ such that, for all $0 \leq j \leq n$, the following hold:\footnote{Here, for notational convenience, we are defining $\ell_{m+1} = \infty$.}
        \begin{enumerate}
            \item if $c_j=\ell_i$ for some $0\leq i\leq m$, then $s'_j \in\cgreencolorA{i}{}$;
            \item if $\ell_i<c_j<\ell_{i+1}$ for some $0\leq i\leq m$ and $c_j$ is an integer, then $s'_j \in\cgreencolorA{i}{+}$;
            \item if $\ell_i<c_j<\ell_{i+1}$ for $0\leq i\leq m$ some but $c_j$ is not an integer and $c_{j'} \geq \ell_{i+1}$ for some $j' < j$, then $s\in\cbluecolorA{i}{-}$; and
            \item if $\ell_i<c_j<\ell_{i+1}$ for some $0\leq i\leq m$ but $c_j$ is not an integer and $c_{j'} < \ell_{i+1}$ for all $j' < j$, then $s\in\credcolorA{i}{\pm}$.
        \end{enumerate}
    \end{claim}

    \noindent
    Note that the constructed PDA has no block cycles, i.e. once a run leaves a block it never comes back. It follows that the claim above implies the run satisfies the formula from \cref{eqn:presb}.

    We now inductively construct $\pi'$ as required. The run $\pi$ must start with state $s_0$ thus with lower bound $0$. Hence, $s'_0$ is chosen to be the corresponding copy from \cgreencolorA{0}{}. As the inductive hypothesis, let us assume that the claim holds for all states $s'_j$ in $\pi'$ for some $0\leq j \leq n$. Now, there are the following cases for the counter value $c_{j+1}$ at state $s_{j+1}$:
    \begin{itemize}
        \item $c_{j+1}=\ell_i$ for some $0\leq i\leq m$. We want to show that $s'_{j+1}$ can be chosen from $\cgreencolorA{i}{}$. From our choice of coefficients, we know there have been no negative updates before, since $c_{j+1}$ is an integer. If $c_j=\ell_i$, we know by induction hypothesis that $s'_j\in\cgreencolorA{i}{}$. We also know that the lower bound on both $c_j$ and $c_{j+1}$ must be at most $\ell_i$ since $\pi$ is a valid run in $\mathcal{A}$. Since the counter update on the transition must have been $+0$ in $\mathcal{A}$, then by construction of $\mathcal{P}$, $s'_{j+1} \in \cgreencolorA{i}{}$ too and there is a transition to it from $s'_j$ in $\mathcal{P}$ that reads $\epsilon$. The only other option for the counter update on the transition from $s_j$ to $s_{j+1}$ in $\mathcal{A}$ is $+1$, that is, $c_j=\ell_i-1$. By induction hypothesis, $s'_j$ must be in $\cgreencolorA{i-1}{}$ (if $\ell_{i-1}=\ell_i-1$) or in $\cgreencolorA{i-1}{+}$ (if $\ell_{i-1}<\ell_i-1$). Then, there is a transition to $s'_{j+1} \in \cgreencolorA{i}{}$ from $s'_j$ since all the $+1$ transitions from $\mathcal{A}$ have copies from $\cgreencolorA{i-1}{}$ and $\cgreencolorA{i-1}{+}$ (reading $a'_{i-1}$) to $\cgreencolorA{i}{}$ in the construction.

        \item $\ell_i<c_{t+1}<\ell_{i+1}$ for some $0\leq i\leq m$ is an integer. We want to show that $s'_{j+1}$ can be chosen from $\cgreencolorA{i}{+}$. Again, from our choice of coefficients, we know that there could not have been any negative updates before. If the counter update on the transition from $s_j$ to $s_{j+1}$ in $\mathcal{A}$ is a $+0$, then $\ell_i<c_j=c_{j+1}<\ell_{i+1}$ is also an integer and, by our hypothesis, $s'_j\in\cgreencolorA{i}{+}$.
        %
        By construction, the $+0$ transition from $s_j$ to $s_{j+1}$ has a copy in $\cgreencolorA{i}{+}$ that leads to $s'_{j+1} \in \cgreencolorA{i}{+}$ reading $\epsilon$. If the counter update on the transition were a $+1$ instead, then $c_j+1=c_{j+1}$ and $s'_j$ would be in $\cgreencolorA{i}{}$ (if $c_j=\ell_i$) or in $\cgreencolorA{i}{+}$ (if $\ell_i<c_j<\ell_{i+1}-1$), and in both cases, by construction, it will have a transition to $s'_{j+1}$ in $\cgreencolorA{i}{+}$ reading $a_i$.
        
        \item $\ell_i<c_{j+1}<\ell_{i+1}$ for some $0\leq i\leq m$ is not an integer and the counter value reached $\ell_{i+1}$ before. We want to show that $s'_{j+1}$ can be chosen from $\cbluecolorA{i}{-}$. Note that the counter update on the transition from $s_j$ to $s_{j+1}$ in $\pi$ could not have been a $+1$ since that would mean that $c_j<\ell_{i+1}-1$ which contradicts the fact that the counter value reached $\ell_{i+1}$ before: Indeed, recall that the absolute sum of all negative updates in $\pi$ is less than $1$. If the counter update on the transition was $+0$, then $\ell_i<c_j=c_{j+1}<\ell_{i+1}$ and $s'_j\in\cbluecolorA{i}{-}$, since the counter must have reached $\ell_{i+1}$ before $s_j$. Both $s_j$ and $s_{j+1}$ must have lower bounds less than $\ell_{i+1}$, that is, at most $\ell_i$ and $\cbluecolorA{i}{-}$ has copies of such transitions with $+0$ updates reading $\epsilon$. If the counter update was $-1$, then $c_j-\delta=c_{j+1}$ and $s'_j$ must be in $\cgreencolorA{i+1}{}$ (if the $-1$ update on the transition from $s_j$ to $s_{j+1}$ is the first negative update in $\pi$) or in $\cbluecolorA{i}{-}$ (if $\ell_i<c_j<\ell_{i+1}$ is not an integer and the $-1$ update on the transition is not the first negative update in $\pi$). Note that $\ell_{i+1}<c_j<\ell_{i+2}$ is not possible, neither is it possible for $c_j$ to be an integer in the latter case, all due to the choice of coefficients. In both of these cases, we see that $-1$ transitions are copied in $\mathcal{P}$ from $\cgreencolorA{i+1}{}$ to $\cbluecolorA{i}{-}$ and within $\cbluecolorA{i}{-}$ reading $b_{i+1}$.

        \item $\ell_i<c_{j+1}<\ell_{i+1}$ for some $0\leq i\leq m$ is not an integer and the counter value never reached $\ell_{i+1}$. We want to show that $s'_{j+1}$ can be chosen from $\credcolorA{i}{\pm}$.
        If the transition from $s_j$ to $s_{j+1}$ has the first negative update in $\pi$, then $\ell_i<c_j<\ell_{i+1}$ must have been an integer. Hence $s'_j$ must have been in $\cgreencolorA{i}{+}$. For the remaining cases, there have been negative update(s) before, $s'_j$ cannot be in a green block $\cgreencolorA{i}{}$ or $\cgreencolorA{i}{+}$, and $c_j$ is not an integer. If $\ell_i<c_j<\ell_{i+1}$, then $s'_j$ is in either $\cbluecolorA{i}{-}$ or in $\credcolorA{i}{\pm}$. The first case is not possible since the run would have passed through $\cgreencolorA{i+1}{}$ before entering $\cbluecolorA{i}{-}$ but we know that the counter value never reached $\ell_{i+1}$. Hence, $s'_j$ is in $\credcolorA{i}{\pm}$ and since all the transitions from $\mathcal{A}$ are copied in this block, $s'_{j+1}$ can be chosen as required. The only remaining case is that $\ell_{i-1}<c_j<\ell_i$ and the counter update on the transition is $+1$. In this case, $s_j$ is either in $\cbluecolorA{i-1}{-}$ or $\credcolorA{i-1}{\pm}$ and we can see that $\mathcal{P}$ has copies of the $+1$ transitions from both of these blocks to $\credcolorA{i}{\pm}$ as required.\qedhere
    \end{itemize}
\end{proof}

\subsection{$k$-reachability for C1PVASS}

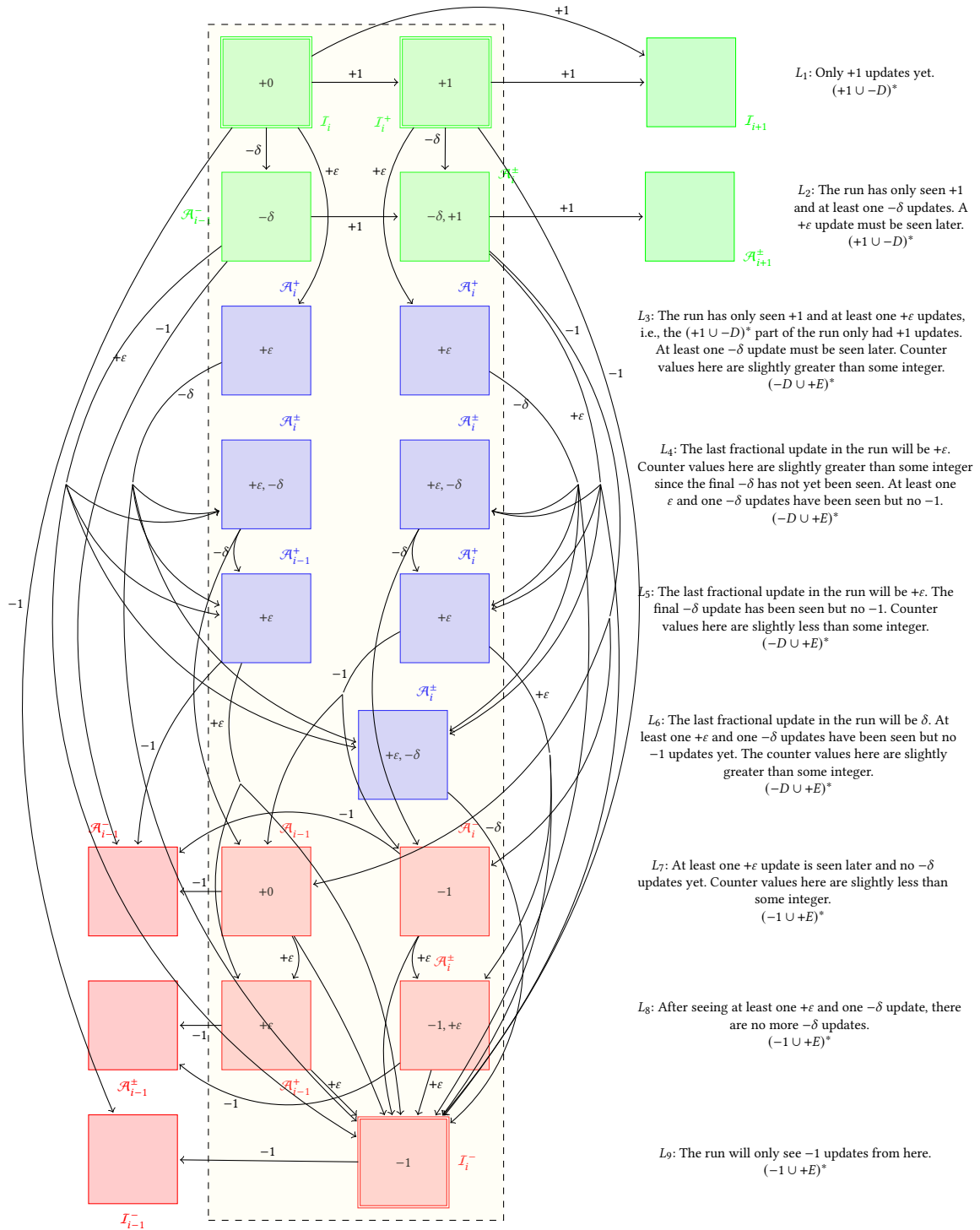
\begin{figure*}[h]
    \centering
    \begin{tikzpicture}[shorten >=1pt,auto,node distance=3.5cm, scale = 0.7, transform shape]
        \tikzstyle{greenblock} = [shape=rectangle, draw=green, minimum size=2cm,fill=green!20]
        \tikzstyle{blueblock} = [shape=rectangle, draw=blue, minimum size=2cm,fill=blue!20]
        \tikzstyle{redblock} = [shape=rectangle, draw=red, minimum size=2cm,fill=red!20]

        
        \node[greenblock,accepting,label={[label distance=.1cm]330:\rgreencolorI{i}{}}](p0)[]{$+0$};
        \node[greenblock,accepting,label={[label distance=.1cm]210:\rgreencolorI{i}{+}}](p1)[right=2cm of p0]{$+1$};
        \node[greenblock](p2)[right=of p1,label={[label distance=.1cm]330:\rgreencolorI{i+1}{}}]{};
        \node[greenblock](d0)[below=1cm of p0,label={[label distance=.1cm]180:\rgreencolorA{i-1}{-}}]{$-\delta$};
        \node[greenblock](d1)[right=2cm of d0,label={[label distance=.1cm]30:\rgreencolorA{i}{\pm}}]{$-\delta,+1$};
        \node[greenblock](d2)[right=of d1,label={[label distance=.1cm]330:\rgreencolorA{i+1}{\pm}}]{};
        \node[blueblock](e0)[below=1cm of d0,label={[label distance=.1cm]80:\rbluecolorA{i}{+}}]{$+\varepsilon$};
        \node[blueblock](e1)[right=2cm of e0,label={[label distance=.1cm]80:\rbluecolorA{i}{+}}]{$+\varepsilon$};
        \node[blueblock](ed0)[below=1cm of e0,label={[label distance=.1cm]80:\rbluecolorA{i}{\pm}}]{$+\varepsilon,-\delta$};
        \node[blueblock](ed1)[right=2cm of ed0,label={[label distance=.1cm]80:\rbluecolorA{i}{\pm}}]{$+\varepsilon,-\delta$};
        \node[blueblock](e0')[below=1cm of ed0,label={[label distance=.1cm]80:\rbluecolorA{i-1}{+}}]{$+\varepsilon$};
        \node[blueblock](e1')[right=2cm of e0',label={[label distance=.1cm]80:\rbluecolorA{i}{+}}]{$+\varepsilon$};
        \node[blueblock](ed)[below right=1.5cm of e0',label={[label distance=.1cm]80:\rbluecolorA{i}{\pm}}]{$+\varepsilon,-\delta$};
        \node[redblock](dmm)[below left=1.5cm of ed,label={[label distance=.1cm]80:\rredcolorA{i-1}{}}]{$+0$};
        \node[redblock](dm0)[right=2cm of dmm,label={[label distance=.1cm]80:\rredcolorA{i}{-}}]{$-1$};
        \node[redblock](dm1)[left=5cm of dm0,label={[label distance=.1cm]100:\rredcolorA{i-1}{-}}]{};
        \node[redblock](edmm)[below=1cm of dmm,label={[label distance=.1cm]280:\rredcolorA{i-1}{+}}]{$+\varepsilon$};
        \node[redblock](edm0)[below=1cm of dm0,label={[label distance=.1cm]above:\rredcolorA{i}{\pm}}]{$-1,+\varepsilon$};
        \node[redblock](edm1)[left=5cm of edm0,label={[label distance=.1cm]below:\rredcolorA{i-1}{\pm}}]{};
        \node[redblock,accepting](n1)[below right=1.5cm of edmm,label={[label distance=.1cm]right:\rredcolorI{i}{-}}]{$-1$};
        \node[redblock](n2)[below=1cm of edm1,label={[label distance=.1cm]below:\rredcolorI{i-1}{-}}]{};

        
        \coordinate[right=2.7cm of e1'](d1t);
        \coordinate[left=3.5cm of ed0](d0t);
        \coordinate[right=2.5cm of ed1](d1t');
        \coordinate[left=2cm of ed0](e0t);
        \coordinate[right=2cm of ed1](e1t);
        \coordinate[above right=2cm of dm1](e0't);
        \coordinate[right=2.3cm of ed](e1't);
        \coordinate[above left=.5cm of ed](e1't');
        \coordinate[above=2cm of e1't'](ed1t);


        \node (Y) [draw=black, fit= (p0) (p1) (n1), inner sep=0.2cm, fill=yellow!20, fill opacity=0.2, dashed] {};


        \node[label={[align=center]right:$L_1$: Only $+1$ updates yet.\\$(+1\cup -D)^*$}](l1)[right=1cm of p2]{};
        \node[label={[align=center]right:$L_2$: The run has only seen $+1$\\and at least one $-\delta$ updates. A\\$+\varepsilon$ update must be seen later.\\$(+1\cup -D)^*$}](l2)[right=1cm of d2]{};
        \node[label={[align=center]right:$L_3$: The run has only seen $+1$ and at least one $+\varepsilon$ updates,\\ i.e., the $(+1\cup -D)^*$ part of the run only had $+1$ updates.\\ At least one $-\delta$ update must be seen later. Counter\\values here are slightly greater than some integer.\\$(-D\cup+E)^*$}](l3)[right=3cm of e1]{};
        \node[label={[align=center]right:$L_4$: The last fractional update in the run will be $+\varepsilon$.\\Counter values here are slightly greater than some integer\\since the final $-\delta$ has not yet been seen. At least one\\$\varepsilon$ and one $-\delta$ updates have been seen but no $-1$.\\$(-D\cup+E)^*$}](l3)[right=3cm of ed1]{};
        \node[label={[align=center]right:$L_5$: The last fractional update in the run will be $+\varepsilon$. The\\final $-\delta$ update has been seen but no $-1$. Counter\\values here are slightly less than some integer.\\$(-D\cup+E)^*$}](l3)[right=3cm of e1']{};
        \node[label={[align=center]right:$L_6$: The last fractional update in the run will be $\delta$. At\\least one $+\varepsilon$ and one $-\delta$ updates have been seen but no\\$-1$ updates yet. The counter values here are slightly\\greater than some integer.\\$(-D\cup+E)^*$}](l4)[right=4cm of ed]{};
        \node[label={[align=center]right:$L_7$: At least one $+\varepsilon$ update is seen later and no $-\delta$\\updates yet. Counter values here are slightly less than\\some integer.\\$(-1\cup+E)^*$}](l5)[right=3cm of dm0]{};
        \node[label={[align=center]right:$L_8$: After seeing at least one $+\varepsilon$ and one $-\delta$ update, there\\are no more $-\delta$ updates.\\$(-1\cup+E)^*$}](l6)[right=3cm of edm0]{};
        \node[label={[align=center]right:$L_9$: The run will only see $-1$ updates from here.\\$(-1\cup +E)^*$}](l7)[right=4.4cm of n1]{};

        
        \path[-]
        (d1)    edge [bend left=30]node[right,align=center,near start]{$-1$}(d1t)
                edge [bend left=20]node[left,align=center,near end]{$+\varepsilon$}(d1t')
        (d0)    edge [bend right=20]node[below,align=center]{$+\varepsilon$}(d0t)
        (e0)    edge [bend right=30]node[below,align=center,near start]{$-\delta$}(e0t)
        (e1)    edge [bend left=20]node[below,align=center,near start]{$-\delta$}(e1t)
        (e0')    edge [bend right=20]node[left,align=center]{$+\varepsilon$}(e0't)
        (e1')    edge [bend left=20]node[right,align=center]{$+\varepsilon$}(e1't)
        (e1')    edge [bend right=20]node[left,align=center,near end]{$-1$}(e1't')
        ;
             
        \path[->]
        (p0)    edge node[above,align=center]{$+1$}(p1)
                edge [bend left]node[above,align=center,near end]{$+1$}(p2)
                edge node[left,align=center]{$-\delta$}(d0)
                edge [bend right]node[left,align=center]{$-1$}(n2)
                edge [bend left=35]node[right,align=center,near start]{$+\varepsilon$}(e0)
        (p1)    edge node[above,align=center]{$+1$}(p2)
                edge node[left,align=center,near start]{$-\delta$}(d1)
                edge [bend left=38]node[right,align=center,near start]{$-1$}(n1)
                edge [bend right=35]node[left,align=center,near start]{$+\varepsilon$}(e1)
        (d0)    edge node[below,align=center]{$+1$}(d1)
                edge [bend right=30]node[left,align=center,very near start]{$-1$}(dm1)
        (d0t)   edge [bend right]node[left,align=center]{}(ed0)
                edge [bend right]node[left,align=center]{}(e0')
                edge [bend right]node[left,align=center]{}(ed)
                edge [bend right=35]node[left,align=center]{}(n1)
        (d1)    edge node[above,align=center]{$+1$}(d2)
        (d1t)   edge [bend left]node[left,align=center]{}(dm0)
                edge [bend left=30]node[right,align=center]{}(dmm)
        (d1t')  edge [bend left]node[left,align=center]{}(ed1)
                edge [bend left]node[left,align=center]{}(e1')
                edge [bend left]node[left,align=center]{}(ed)
                edge [bend left=25]node[left,align=center]{}(n1)
        (e0t)   edge [bend right]node[left,align=center]{}(ed0)
                edge [bend right]node[left,align=center]{}(e0')
                edge [bend right]node[left,align=center]{}(ed)
                edge [bend right=27]node[left,align=center]{}(n1)
        (e1t)   edge [bend left]node[left,align=center]{}(ed1)
                edge [bend left]node[left,align=center]{}(e1')
                edge [bend left]node[left,align=center]{}(ed)
                edge [bend left=20]node[left,align=center]{}(n1)
        (ed0)   edge [bend right]node[left,align=center]{$-\delta$}(e0')
                edge [bend right=30]node[left,align=center]{}(dmm)
        (ed1)   edge [bend right]node[left,align=center]{$-\delta$}(e1')
                edge [bend right=30]node[left,align=center]{}(dm0)
        (e0')   edge [bend right=20]node[left,align=center]{$-1$}(dm1)
        (e0't)  edge [bend right=25]node[left,align=center]{}(edmm)
                edge [bend left=20]node[left,align=center]{}(n1)
        (e1't)  edge [bend left=20]node[left,align=center]{}(edm0)
                edge [bend left=20]node[left,align=center]{}(n1)
        (e1't') edge [bend right=20]node[left,align=center]{}(dmm)
                edge [bend right=20]node[left,align=center]{}(dm0)
        (ed)    edge [bend left=50]node[right,align=center,very near start]{$-\delta$}(n1)
        (dmm)   edge node[above,align=center]{$-1$}(dm1)
                edge [bend left]node[left,align=center]{$+\varepsilon$}(edmm)
                edge [bend left=5]node[left,align=center]{}(n1)
        (dm0)   edge [bend right=40]node[above,align=center,near start]{$-1$}(dm1)   edge [bend right]node[right,align=center]{$+\varepsilon$}(edm0)
                edge [bend right=22]node[right,align=center]{}(n1)
        (edmm)  edge node[below,align=center]{$-1$}(edm1)
                edge node[above,align=center]{$+\varepsilon$}(n1)
        (edm0)  edge [bend left=40]node[below,align=center,near end]{$-1$}(edm1)
                edge node[above right,align=center]{$+\varepsilon$}(n1)
        (n1)    edge node[above,align=center]{$-1$}(n2)
        ;
        \end{tikzpicture}
        \caption{The $i^\text{th}$ slice of the PDA $\mathcal{P}$ constructed for $k$-reachability of a C1PVASS. The slice itself is inside the dashed box, the text to the right provides intuition for the layer. The subscript being $i$ for $0\leq i\leq m$ of a block (eg, $i$ in \cgreencolorA{i}{}) denotes that all the states in the block have lower bounds at most $\ell_i$. Note $+0$ is omitted unless it is the only option for transitions in the block.}
    \label{fig:C1PVASS-reach-to-PDA}
\end{figure*}

In this section, we show that $k$-reachability for $k\in\mathbb{N}$ is also in \textsc{NP} for C1PVASS. Unlike in 0-C1PVASS, $k$-coverability does not imply $k$-reachability in C1PVASS, since scaling down the vectors along the entire run can lead to some lower bounds being violated. For instance, consider a C1PVASS with 3 states, namely, $s_0$, $s_1$ and $f$ with $\ell(s_1)=1$ and a $+1$ update on both $s_0\rightarrow s_1$ and $s_1\rightarrow f$. For any accepting run $\pi$, the counter value at $s_1$ must be $1$. This means that even if the second $+1$ update is scaled down, the counter value at $f$ must be strictly greater than $1$. Hence, for this C1PVASS, $1$-coverability holds but $1$-reachability does not.

Like in the previous section, we construct a PDA and a Presburger formula such that the PDA accepts a word that satisfies the Presburger formula iff the C1PVASS has a $k$-reaching run, for some given $k\in\mathbb{N}$. However, the construction is not as straightforward as in the previous section since it is not always possible to reach a specific counter value by scaling down all negative counter updates to arbitrarily small numbers. Instead, we first introduce a normal form of scaled runs (i.e. the sequences of coefficients) that guides our construction a PDA with no block cycles in the same way \cref{lem:extremal_runs} guided our construction for reachability.

\subsubsection{The Dense Normal Form (DNF)}
We show that all $k$-reaching runs have a normal form which scales the counter updates in the run so that the positive updates are concentrated towards the start of the run and the negative updates towards the end of the run. Formally, let $\pi$ be a run in the C1PVASS which reaches the counter value $k\in\mathbb{N}$. Let $P_\pi$ be the sum of all positive updates in $\pi$ and $N_\pi$ be the sum of all negative updates. Define $I_P^\pi=\lfloor P_\pi\rfloor$, $F_P^\pi=P_\pi-I_P^\pi$, $I_N^\pi=\lceil N_\pi\rceil$ and $F_N^\pi=I_N^\pi-N_\pi$. Clearly $I_P^\pi + F_P^\pi + F_N^\pi + I_N^\pi = k$ and, moreover, $I_P^\pi + I_N^\pi = k$ and $F_P^\pi = - F_N^\pi$ since $k$ is an integer. Our intention, to define a normal form and argue all runs can be put in it, is to scale the first $I_P^\pi$ positive updates and the last $I_N^\pi$ negative updates along the run in full, i.e. $\gamma = 1$. The remaining positive and negative updates can be scaled arbitrarily (small) as long as their sum adds up to $0$. For the latter, we will scale down positive and negative updates by coefficients from $E$ and $D$ respectively, where $E$ and $D$ are finite sets of arbitrarily small ``epsilons'' ($\varepsilon$) and ``deltas'' ($\delta$).

\begin{lemma}\label{thm:normal-form}
    Let $\pi$ be a run in the C1PVASS which reaches the counter value $k\in\mathbb{N}$. Then, there exists a run $\pi'$ such that:
    \begin{enumerate}
        \item\label{item:normalform-1} The sequence of states in $\pi'$ is the same as in $\pi$, and
        \item\label{item:normalform-2} $\pi'$ is also a $k$-reaching run.
        \item\label{item:normalform-3} All positive and negative updates in $\pi'$ are scaled from the set $\Set{1}\cup E$ and $\Set{1}\cup D$, respectively.
        \item\label{item:normalform-4} The sequence of all nonzero updates in $\pi'$ is of the form $(\Set{+1}\cup -D)^*(-D\cup+E)^*(\Set{-1}\cup+E)^*$. 
        \item\label{item:normalform-5} Let $\restr{\pi'}{E \cup D}$ be the sequence of counter updates restricted to $+E$'s and $-D$'s. Then,  $\restr{\pi'}{E \cup D}$ is of the form:
        \begin{itemize}
            \item \label{item:type1} $(-D)^*(+E)(+E \cup -D)^*(-D)$, in which case, for any proper prefix of $\restr{\pi'}{E \cup D}$ with at least one epsilon, the sum of all epsilons and deltas is positive; or
            \item \label{item:type2} $(+E \cup -D)^*(-D)( +E)^+$, and for any proper prefix of $\restr{\pi'}{E \cup D}$ that contains an epsilon and not the final delta, the sum of all epsilons and deltas is positive while for any  proper prefix with the final delta, the sum is negative.
        \end{itemize}
        \item\label{item:normalform-6} For configurations $q=(s,\alpha,c)$ in $\pi$ and $q'=(s,\alpha,c')$ in $\pi'$ which occur in the same position in both runs, $c'\geq\lfloor c\rfloor$.
    \end{enumerate}
\end{lemma}


Note that the DNF (in particular
\cref{item:normalform-4}) precludes $-1$ updates before a $+1$ update. Indeed, if this ever happens then one can scale both updates down while preserving the final counter value to obtain a run in DNF. The last technical hurdle in transforming a run into DNF is ensuring \ref{item:normalform-5}. That epsilons and deltas exist as required in both cases should be clear. Finally, convincing oneself that the two cases are exhaustive is easy since there is a delta iff there is an epsilon in the original run because otherwise the fractional parts cannot cancel each other out.

\begin{remark}\label{rem:eps-iff-del}
    A run $\pi$ of a C1PVASS that reaches an integral counter value has an $+\varepsilon$ update if and only if it has a $-\delta$ update.
\end{remark}

\begin{example}
    Let the sequence of nonzero updates in a 
    $1$-reaching run be $+1,+0.8,-0.9,-0.9,+1,-1,+1$. This run is not in the dense normal form. To transform it into the normal form, we choose different values of $\gamma$ to scale the updates along the run to obtain $+1,+1,-\delta_1,-\delta_2,+\varepsilon_1,-1,+\varepsilon_2$. It is easy to see that this is also a $1$-reaching run and the integer parts of the counter values along the run are at least those along the original run.
\end{example}

Since we know that if a $k$-reaching run exists, then a $k$-reaching run in dense normal form exists, we construct a PDA and a Presburger formula which simulate runs of the C1PVASS in DNF.

\subsubsection{Constructing the PDA $\mathcal{P}$}
The construction of $\mathcal{P}$ is shown in \cref{fig:C1PVASS-reach-to-PDA}. Each slice consists of 9 layers, namely $L_1,\dots,L_9$, each of which contains one or two blocks. Each block, like in \cref{sec:cover-C1PVASS} is a copy of the C1PVASS $\mathcal{A}$ restricted to states with lower bounds at most $\ell_i$, where $i$ is the subscript of the block, and transitions with updates written inside the block. The transitions labelled $+\varepsilon$ and $-\delta$ correspond to $+1$ and $-1$ update transitions respectively, but we label them differently to make the explanation later easier to read. Each block also has all transitions from $\mathcal{A}$ with a $+0$ update. Note that the connections between blocks allow exactly the transitions on the labels and not the $+0$ transitions from $\mathcal{A}$. 

In the figure, we are assuming that $\ell_{i-1}<\ell_i$. If $\ell_{i-1}=\ell_i-1$, there will be the following changes:
\begin{enumerate}
    \item All transitions entering \rredcolorA{i-1}{-} will now go to \rredcolorA{i-2}{-} instead (both in $L_7$), and
    \item All transitions entering \rredcolorI{i-1}{-} will now go to \rredcolorI{i-1}{} instead.
\end{enumerate}
Checking whether $\ell_{i-1}=\ell_i-1$ can be done beforehand for all $0<i\leq m$, and $\mathcal{P}$ can be constructed accordingly. The discussion that follows still holds.

 Due to the complexity of the PDA, providing a formal proof like in \cref{sec:cover-C1PVASS} is not feasible. Rather, we provide intuition behind the construction of the PDA and the Presburger formula, and the reader should convince themselves that the construction is correct.

\begin{remark}
    By construction, there are no block cycles in $\mathcal{P}$. That is, once a run exits a block, it cannot enter the same block later. This follows due to the simple observation that, from every block, there are transitions only to a block either to the left or to the right on the same level (but never both), or a block in a lower layer.
\end{remark}

Note that, since there is no way to enter a block in an upper layer from a lower one, and a run always starts from \rgreencolorI{0}{} (the leftmost green block), any run will first traverse green blocks, moving to the right, then it will either enter a blue block and move down or directly enter a red block and start moving to the left. 

\begin{lemma}\label{lem:C1PVASS_reach_run_split}
    The sequence of states in any run in $\mathcal{P}$ is a sequence of states in green blocks, followed by a (possibly empty) sequence of states in a red block, followed by a (possibly empty) sequence of states in red blocks. Furthermore, while the run is visiting green, blue and red blocks, the index of the slices is non-increasing, constant and non-increasing, respectively.
\end{lemma}
The proof follows by construction.

We now show how the PDA $\mathcal{P}$ is split into 3 main components and the letters read on the transitions in slice $i$.
\begin{itemize}
    \item The \textcolor{green}{green} component consists of the first two layers. This component corresponds to the $(+1\cup -D)^*$ part of the run in \cref{item:normalform-4}. This is easy to see since this component looks exactly like \cref{fig:C1PVASS-lb-to-PDA}, with the exception that the states in the second layer (which were the states in blue and red states in \cref{fig:C1PVASS-lb-to-PDA}), are not accepting. The PDA reads alphabet $a_i$ on all $+1$ transitions in and to green blocks in \cref{fig:C1PVASS-reach-to-PDA}, except the transitions entering the next slice, on which it reads $a_i'$, and $d_i$ on all the $-\delta$ transitions
    \item The \textcolor{blue}{blue} component consists of layers 3, 4, 5 and 6 which correspond to the $(+E\cup-D)^*$ part of the run. This component is entered after reading the first $+\varepsilon$ update. Due to \cref{item:normalform-5}, the run stays in a single slice during this part of the run. Note that the run cannot go back to a green block after entering a blue block. The empty letter $\epsilon$ is read on all transitions in and to this component.
    \item The \textcolor{red}{red} component, consisting of the last 3 layers, corresponds to the $(-1\cup+E)^*$ part of the run, and is entered after the first $-1$ or the last fractional update. Note that the run can never exit this component once it is entered. The PDA reads $b_i'$ on all the $-1$ transitions in and to $L_7$, $b_i$ on all transitions in and to $L_8$ and $L_9$, and $\epsilon$ on all other transitions.
\end{itemize}

Now that we have established what the 3 components are, we provide intuition on blocks and transitions within each layer in slice $i$.
\begin{itemize}
    \item[\textcolor{green}{$L_1$}:] Since this layer works similarly to \cref{fig:C1PVASS-lb-to-PDA} on $+1$ and $-\delta$ updates, we only study what it does on $+\varepsilon$ and $-1$ updates. When a $+\varepsilon$ update is read on this layer, the run enters layer 3, since there have only been $+1$ updates before. If the run sees a $-1$ update in this layer, it must enter the red component, and due to \cref{rem:eps-iff-del}, it cannot see any $+\varepsilon$ edges in the red component either, so it enters the last layer.
    \item[\textcolor{green}{$L_2$}:] Similar to layer 1, we only study what the run does on $+\varepsilon$ and $-1$ updates from this layer. On seeing a $+\varepsilon$ update, it enters $L_4$ if the run is of the second type in \cref{item:normalform-5} and the run still has to see more deltas, $L_5$ if again, the run is of the second type in \cref{item:normalform-5} and a $-\delta$ from the green component is the last delta in the run, $L_9$ if there are no more fractional updates, and $L_6$, if the run is of the first type in \cref{item:normalform-5}.
    \item[\textcolor{blue}{$L_3$}:] When the run is in this layer, the run must have only seen $+\varepsilon$ and $+1$ update before. The run cannot immediately read a $-1$ immediately from this block since, by \cref{rem:eps-iff-del}, there must be a $-\delta$ update later in the run and that cannot happen after a $-1$ update since the run is in DNF. On reading a delta, similar to $L_2$, the run enter $L_4$, $L_5$, $L_6$ or $L_9$ depending on what type of a run it is, and whether the $-\delta$ was the last delta or the last fractional update in the run.
    \item[\textcolor{blue}{$L_4$}:] On the last $-\delta$ read after entering this block (which there must be, since otherwise the run would have entered $L_5$ instead of $L_4$), the run enters $L_5$ or $L_7$, which, as we will see, ensure that the the last fractional update in the run is $+\varepsilon$.
    \item[\textcolor{blue}{$L_5$}:] This layer is entered when at least one $-\delta$ is seen before, and it ensures that the only (and at least one) fractional updates seen later are $+\varepsilon$. On a $-1$ update, the run enters $L_7$, which ensures this property too. On a $+\varepsilon$ update, it enters $L_8$ if there are more $+\varepsilon$ updates to come, or $L_9$ if this was the last $+\varepsilon$ update of the run.
    \item[\textcolor{blue}{$L_6$}:] The run enters this layer only if there was at least one $+\varepsilon$ and one $-\delta$ updates before, and it ensures that the last fractional update will be a $-\delta$ by only going to $L_9$ on the last $-\delta$, hence, avoiding all epsilon updates that could occur in the $(-1\cup +E)^*$ portion of the run.
    \item[\textcolor{red}{$L_7$}:] The run enters this layer after seeing at least one $-\delta$ update previously, and the only way to reach an accepting state is by seeing a $+\varepsilon$ and entering $L_8$ or $L_9$.
    \item[\textcolor{red}{$L_8$}:] All runs that enter this layer have seen at least one $+\varepsilon$ and one $-\delta$ updates before. The run stays in this layer on seeing $+\varepsilon$ and $-1$ updates, until it enters $L_9$ on the last $\varepsilon$.
    \item[\textcolor{red}{$L_9$}:] The run enters this layer after having seen all the fractional updates (if any), and only sees $-1$ updates henceforth. 
\end{itemize}

\subsubsection{The Presburger formula $\varphi$}

The main intuition for having the Presburger formula is to make sure that the run stays in the correct block. For example, on reading a $+1$ in \rgreencolorI{i}{+}, there is a choice to either stay within the block or move to \rgreencolorI{i+1}{}. The formula ensures that it stays in \rgreencolorI{i}{+} when the counter value is in the interval $(\ell_i,\ell_{i+1})$ and moves to \rgreencolorI{i+1}{} when the counter value reaches $\ell_{i+1}$.

We will also use a formula to ensure that the sum of all the $+1$ and $-1$ updates is exactly $k$.

Using \cref{lem:C1PVASS_reach_run_split}, we are able to split the Presburger formula into 3 conjuncts as well.

For the \textcolor{green}{green} component, note that the PDA restricted to the first 2 layers looks exactly like \cref{fig:C1PVASS-lb-to-PDA}, except the accepting states in the second layer are not accepting, but this is only because the counter value is not an integer in this layer. The Presburger formula for these also is similar to the one for the coverability PDA.
\begin{equation}\label{formula:reach-green}
    \begin{aligned}
        \varphi_G=\bigwedge_{k=1}^m\left(
            \begin{aligned}
                \big(\#_{a_{k-1}'}&=0\big)\vee\\
                \bigg(\Big(\sum_{i=0}^{k-1}\#_{a_i}+\#_{a_i'}\geq \ell_k\Big)&\wedge\Big(\sum_{i=0}^{k-1}\#_{d_i}=0\Big)\bigg)\vee\\
                \bigg(\Big(\sum_{i=0}^{k-1}\#_{a_i}+\#_{a_i'}>\ell_k\Big)&\wedge\Big(\sum_{i=0}^{k-1}\#_{d_i}>0\Big)\bigg)
            \end{aligned}
        \right)
    \end{aligned}
\end{equation}

For the \textcolor{blue}{blue} component in a slice $i$, the fact that the only updates here are $+\varepsilon$ and $-\delta$ means that the counter value never goes below $\ell_i$ in the blue blocks in slice $i$, unless the run is of the second type in \cref{item:normalform-5} and the run has seen exactly $\ell_i$ many $+1$ updates before. In which case, on seeing the final delta, it enters a copy of a blue block from the slice $i-1$ (the left block in $L_5$). Hence, one can check that the run always satisfies the lower bounds of the states in the blocks when in the blue component.

For the \textcolor{red}{red} component, we look at $L_7$ and $L_8$ first. When entering blocks in these layers, an accepting run will have seen at least one $-\delta$ before, and must have at least one $+\varepsilon$ yet to be seen. Which means that the counter value here is slightly lower than the total number of $+1$ updates subtracted by the number of $-1$ updates seen before. The run reads $b_i'$ when entering and in this block, which gives us the following formula.
\begin{equation}
    \begin{aligned}
        \varphi_R'=\bigwedge_{k=1}^m\left(                      \Big(\#b_i'=0\Big)\vee\bigg(\sum_{i=1}^{m}\#_{a_i}+\#_{a_i'}-\sum_{j=k}^{m}\#_{b_j'}>\ell_k\bigg)
        \right)
    \end{aligned}
\end{equation}
The first conjunct considers the case where the run never enters slice $k$. When it does, the second conjunct ensures the number of $+1$ updates seen before is strictly greater than the number of $-1$ updates seen before entering these layers.

We now consider $L_9$. To enter these blocks, either no epsilons and deltas have been seen before, or all $+\varepsilon$ and $-\delta$ updates have already been seen. Thus, using the fact that the sum of all fractional updates is exactly $0$, the counter update in the states in these blocks will be exactly the total number of $+1$ updates subtracted by the number of $-1$ updates seen before. The run reads $b_i$ when entering and in this block, which gives us the following formula.
\begin{equation}
    \begin{aligned}
        \varphi_R=\bigwedge_{k=1}^m\left(                      \Big(\#b_i=0\Big)\vee\bigg(\sum_{i=1}^{m}\#_{a_i}+\#_{a_i'}-\sum_{j=k}^{m}\#_{b_j'}+\#_{b_j}\geq\ell_k\bigg)
        \right)
    \end{aligned}
\end{equation}
Like the previous case, the first conjunct checks if the run enters the last layer in the $k^\text{th}$ slice, if so, the second conjunct ensures that the counter update, which is the number of $b$'s subtracted by the number of $a$'s here, is at least the lower bound on that block.

Finally, we have a formula to ensure that the total number of $b$'s subtracted by the number of $a$'s (i.e., sum of all the $+1$ and $-1$ updates in $\mathcal{A}$) in the run is exactly $k$.
\begin{equation}\label{eqn:pres-reach}
    \begin{aligned}
        \varphi_k=\sum_{i=0}^m \#_{a_i}+\#_{a_i'}-\#_{b_i}-\#_{b_i'}=k
    \end{aligned}
\end{equation}

The final formula $\varphi$ will be a conjunction of all the formulas described above.

\begin{equation*}
    \varphi=\varphi_G\wedge\varphi_R'\wedge\varphi_R\wedge\varphi_k
\end{equation*}

Now, with the PDA $\mathcal{P}$ and formula $\varphi$, we get the main result of this section.

\begin{theorem}\label{thm:C1PVASS+lb_reach_to_PDA+Pres}
    For all runs $\pi$ in $\mathcal{A}$, there is one in $\mathcal{P}$ with the same sequence of states whose Parikh image satisfies the Presburger formula $\varphi$ from \cref{eqn:pres-reach} if and only if $\pi$ is $k$-reaching in $\mathcal{A}$.
\end{theorem}

The proof is similar to that of \cref{thm:C1PVASS+lb_to_PDA+Pres}, simulating the $k$-reaching runs of $\mathcal{A}$ into accepting runs in $\mathcal{P}$ satisfying $\varphi$ and vice-versa. We leave the details to the reader.

\subsection{Boundedness for C1PVASS}

For boundedness, we first note that the set of all reachable counter values from a C1PVASS $\mathcal{A}$ is a subset of all reachable values of the 0-C1PVASS $\mathcal{A'}$, where $\mathcal{A'}$ is obtained by replacing the lower bounds in $\mathcal{A}$ by $0$. It follows, from \cref{thm:zboundedness-ptime}, that $2^{(4|\mathcal{A}|)^9}$ is an upper bound on the largest reachable counter value in $\mathcal{A}$, if $\mathcal{A}$ is bounded, where $|\mathcal{A}|$ is the size of the encoding of $\mathcal{A}$. Hence, we can ask for $k$-coverability with $k=2^{(4|\mathcal{A}|)^9}$ and if the answer is yes, the set of all reachable counter values is unbounded. Hence, boundedness is in \textsc{coNP}.

\begin{theorem}
    Deciding boundedness of a C1PVASS is in \textsc{coNP}.
\end{theorem}

\section{Conclusion}

In this work we established reachability, coverability, and boundedness are decidable in polynomial time for continuous PVASS in one dimension (C1PVASS). When the model is extended with lower-bound guards for the counter on the states, we proved reachability and coverability are in \textsc{NP} while boundedness is in \textsc{coNP}. No lower bounds for the latter seem to follow from the literature and in fact we conjecture that they are also decidable in polynomial time.
In the direction of using C1PVASS as approximations of PVASS, we posit the most
interesting direction is to add both upper and lower bounds to the values the counter can take (cf. \cite{coca}) towards
an approximation of one-counter pushdown automata.


\bibliographystyle{ACM-Reference-Format}
\bibliography{refs}



\end{document}